\documentclass[a4paper,fleqn]{article}
\pdfoutput=1
\usepackage[ngerman,american]{babel}  
\usepackage[utf8]{inputenc}
\usepackage{geometry}
\usepackage{array}
\usepackage{float}
\usepackage[urlcolor=blue]{hyperref}
\usepackage{dsfont} 
\usepackage{subcaption} 
\usepackage{dutchcal}

\usepackage{latexsym}                                      
\usepackage{tikz}       
\usepackage{graphicx}  
\usepackage{bm}  
\usepackage{mathtools, amssymb,amsthm}     
\usepackage[]{mdframed}
\usepackage[round]{natbib}
\usepackage[shortlabels]{enumitem}
\usepackage{pdfpages}
\usepackage[]{algorithm2e}
\usepackage{pdflscape}
\usepackage{xspace}
\usepackage{setspace}
\usepackage{xcolor}
\usepackage{diagbox}
\usepackage{booktabs} 
\newenvironment{nobreaks}{\vbox\bgroup}{\egroup}

\definecolor{airforceblue}{rgb}{0.36, 0.54, 0.66}
\definecolor{darkergreen}{rgb}{0.0, 0.5, 0.0}
\definecolor{plotorange}{rgb}{0.9290, 0.6940, 0.1250}
\definecolor{plotblue}{rgb}{0, 0.4470, 0.7410}
\definecolor{plotgreen}{rgb}{0.4660, 0.6740, 0.1880}
\hypersetup{
    colorlinks,
    linkcolor={darkergreen},
    citecolor={airforceblue},
    urlcolor={blue!80!black}
}
\title{seMCD: Sequentially implemented Monte Carlo depth computation with statistical guarantees}

\date{}
\author{Felix Gnettner\footnote{Institut für Mathematische Stochastik, Fakultät für Mathematik, Otto-von-Guericke-Universit\"at Magdeburg, Germany. E-mail: felix.gnettner@ovgu.de}, \ Claudia Kirch\footnote{Institut für Mathematische Stochastik, Fakultät für Mathematik, Otto-von-Guericke-Universit\"at Magdeburg, Germany. E-mail: claudia.kirch@ovgu.de}, \ Alicia Nieto-Reyes\footnote{Departamento de Matemáticas, Estadística y Computación, Facultad de Ciencias, Universidad de Cantabria, Spain. E-mail: alicia.nieto@unican.es}}
\makeatletter
\newcommand{\V}{\mathbb{V}\mathrm{ar}}

\newcommand{\E}{\mathbb{E}}
\renewcommand{\P}{\mathbb{P}}

\newcommand{\deq}{\overset{\mathcal D}{=}}

\newcommand{\bx}{\bm{x}}
\newcommand{\by}{\bm{y}}
\newcommand{\bZ}{\bm{Z}}
\newcommand{\bz}{\bm{z}}
\newcommand{\sB}{\mathfrak{B}}

\renewcommand{\le}{\leqslant}
\renewcommand{\ge}{\geqslant}
\renewcommand{\leq}{\leqslant}
\renewcommand{\geq}{\geqslant}

\newcommand*\rel@kern[1]{\kern#1\dimexpr\macc@kerna}
\newcommand*\widebar[1]{
  \begingroup
  \def\mathaccent##1##2{
    \rel@kern{0.8}
    \overline{\rel@kern{-0.8}\macc@nucleus\rel@kern{0.2}}%
    \rel@kern{-0.2}
  }
  \macc@depth\@ne
  \let\math@bgroup\@empty \let\math@egroup\macc@set@skewchar
  \mathsurround\z@ \frozen@everymath{\mathgroup\macc@group\relax}%
  \macc@set@skewchar\relax
  \let\mathaccentV\macc@nested@a
  \macc@nested@a\relax111{#1}%
  \endgroup
}

\newcommand*\wideunderbar[1]{%
  \begingroup
  \def\mathaccent##1##2{%
    \rel@kern{-0.2}%
    \underline{\rel@kern{0.2}\macc@nucleus\rel@kern{-0.8}}%
    \rel@kern{0.8}%
  }%
  \macc@depth\@ne
  \let\math@bgroup\@empty \let\math@egroup\macc@set@skewchar
  \mathsurround\z@ \frozen@everymath{\mathgroup\macc@group\relax}%
  \macc@set@skewchar\relax
  \let\mathaccentV\macc@nested@a
  \macc@nested@a\relax111{#1}%
  \endgroup
}
\makeatother
\theoremstyle{plain}
\newtheorem{Theorem}{Theorem}[section]

\newtheorem{Remark}[Theorem]{Remark}
\newtheorem{Lemma}[Theorem]{Lemma}

\newtheorem{Corollary}[Theorem]{Corollary}

\expandafter\let\expandafter\oldproof\csname\string\proof\endcsname
\let\oldendproof\endproof
\renewenvironment{proof}[1][\proofname]{\oldproof[\bfseries #1]}{\oldendproof}

\newtheorem{Condition}{Condition}

\begin{document}
\maketitle
\begin{abstract}
    \noindent    
Statistical depth functions provide center-outward orderings in  spaces of dimension larger than one, where a natural ordering does not exist.  The numerical evaluation of such depth functions can be computationally prohibitive, even for relatively low  dimensions. We present a novel \textbf{se}quentially implemented \textbf{M}onte \textbf{C}arlo methodology for the computation of, theoretical and empirical, \textbf{d}epth functions and related quantities (seMCD), that  outputs an interval, a so-called seMCD-bucket, to which the quantity of interest belongs with a high probability prespecified by the user. For specific classes of depth functions, we adapt algorithms from sequential testing, providing finite-sample guarantees. For depth functions dependent on unknown distributions, we offer asymptotic guarantees using non-parametric statistical methods. In contrast to plain-vanilla Monte Carlo methodology the number of samples required in the algorithm is random but typically much smaller than standard choices suggested in the literature. The seMCD method can be applied to various depth functions, covering multivariate and functional spaces. We demonstrate the efficiency and reliability of our approach through empirical studies, highlighting its applicability in outlier or anomaly detection, classification, and depth region computation. In conclusion, the seMCD-algorithm can achieve accurate depth approximations with few Monte Carlo samples while maintaining rigorous statistical guarantees.

\end{abstract}


\section{Introduction}
The natural ordering on the real line forms the basis for defining order statistics, quantiles, ranks and further non-parametric descriptive key values. In multivariate spaces, functional spaces or in a setting where object data are of interest, a natural order does not generally exist. Statistical depth notions \citep{ZuoSerfling, Alicia, AliciaLuisTukey} were conceived to fill in this void, offering a variety of center-outward orderings that entail e.g.\ rank-based inference \citep{LiuSingh}. Recurrent applications of depth are 
\begin{enumerate}
    \item[(i)] outlier or anomaly detection \citep[for multivariate data]{ChenDang2010,DangSerfling2010,mozharovskyi2024anomalydetectionusingdata},  \citep[for functional data]{FebreroGaleano2008,Hubert2015} and \citep[for object data like gene expression data]{Alicia2022},
    \item[(ii)] classification  \citep[for multivariate data]{Joernsten2004,GhoshChaudhuri2005,LiCuestaLiu2012,Lange2014} and  \citep[for functional data]{Cuevas2007Projection,Sguera2014},  and
    \item[(iii)] depth region computation \citep{Cascos2007,Mosler2013}. 
\end{enumerate}

These applications require computing the empirical depth of elements in the underlying space. Exact computation of empirical depth functions  is often costly or even impossible due to complex analytical calculations. For some well-known depth functions, however, computationally efficient algorithms for the exact computation have been proposed, e.g., \cite{Dyckerhoff1} for the empirical Tukey depth \citep{Tukey}. For computing the empirical simplicial depth \citep{LiuSimplicial} of a single point $\bz \in \mathbb{R}^d$ with respect to a sample of size $n$, the naive approach has computational complexity $O(n^{d+1})$, where the constants also depend on $d$. There exist more advanced exact algorithms in the literature with a computational complexity of $O(n \cdot \log n)$  for $d=2$ \citep[Theorem 3.2]{Aloupis2002}, $O(n^2)$ for $d=3$ and $O(n^4)$ for $d=4$ \citep{Cheng2001}. \cite{Afshani2016} prove that exact computation of the empirical simplicial depth is P-complete and W[1]-hard, if the dimension $d$ is part of the input. Moreover, they present an algorithm with complexity $O(n^d \log(n))$ for the case $d \geq 5$, which has been the first improvement on the trivial bound of $O(n^{d+1})$. For many other depth functions, e.g.\ the functional band depth and its modifications \citep{LR2009} (theoretical and with respect to large data sets) or the simplicial volume depth \citep{Oja1983,ZuoSerfling}, there are still no fast algorithms at all.

Monte Carlo approximations offer a remedy for these computational problems. The calculation of many empirical depth functions is based on sums that require the evaluation of a large number of summands, e.g.\ $O(n^{d+1})$ in the case of the simplicial depth. In such situations, the empirical depth function is often approximated by Monte Carlo simulations that terminate after a pre-specified number of iterations. However, the associated uncertainty inherent in these Monte Carlo methodologies is usually not taken into account. This may cause problems when it comes to making decisions based on the so-obtained values. Two recent exceptions are given in \cite{BaharavLai2023} and \cite{Briend2023}. The former noted that identifying the empirical simplicial median corresponds to the best arm identification problem with multi-armed bandits and provided an algorithm with probabilistic guarantees. The latter investigated the pointwise stochastic approximation quality of the random Tukey depth \citep{AliciaCuesta} as a point estimator for the Tukey depth.

Uncertainty quantification is also key when applying Monte Carlo methods to compute the outcome of statistical tests -- for example when test decisions are based on $p$-values obtained via bootstrapping or permutation methods. There, it is not necessary to compute the exact $p$-value, but it suffices to distinguish between $p<\alpha$ and $p>\alpha$.  \citet{Gandy2009,Gandy2020a,Gandy2020b} make use of this fact by providing a sequential Monte Carlo technique that can distinguish between these two scenarios. Their methodology comes with the anytime-valid guarantee that the probability of making the wrong decision between these two scenarios is smaller than a user-specified value. At the same time, these techniques typically require much fewer Monte Carlo samples than usually suggested in such contexts. 

Similarly, in the context of depth-related calculations, it often suffices to know the depth region instead of the precise value or a point estimation. In outlier or anomaly detection, it is enough to know that the depth value is smaller than e.g.\ $0.05$. In maximum depth classification, it is only necessary to get evidence that one depth value is smaller than the other. Inspired by the above papers,  we propose sequential Monte Carlo techniques  in order to provide new algorithms that output such depth information in all those situations and more. Our approach works as long as the quantity of interest can essentially be calculated as an expected value, or a monotonic transformation of such. We call it \textbf{se}quential \textbf{M}onte \textbf{C}arlo \textbf{D}epth (seMCD) methodology, which offers two main advantages: First, it typically requires fewer Monte Carlo samples than are usually suggested in the literature (such as $10^5$ as e.g.\ suggested by \citet[Example 4]{Ramsay2019}). Secondly, it comes with the guarantee that the decision is correct with a probability of at least $1-\alpha$, where $0<\alpha<1$ is a \emph{tolerance} parameter pre-specified by the user. In some situations, such a guarantee is given in a suitable asymptotic sense. 
 
The seMCD-methodology outputs a set of values that contains the quantity of interest with a user-chosen probability. This output set is denoted as a seMCD-bucket. In contrast to confidence sets, the finite family of possible output seMCD-buckets is also a user-specified input to the algorithm. While the methodology also gives a point value of the quantity of interest, the statistical guarantee is only attached to the bucket. Essentially, Monte Carlo samples are taken until this guarantee is achieved. Roughly speaking, the smaller the input buckets, the more Monte Carlo samples are required on average. We will also motivate the benefits of overlapping seMCD-buckets as input to the algorithm.

The paper is organized as follows: In Section~\ref{section_2_depth}, we introduce the class of depth functions and their applications, for which the seMCD-methodology of this paper can be  used. In Section~\ref{sec_main_seMCD} we explain in detail the algorithms behind the seMCD-methodology and how their statistical properties depend on properties of the chosen boundary sequences. This is followed, in Section~\ref{sec_boundary}, by a discussion of the construction of such boundary sequences, both parametrically in the Bernoulli-case and non-parametrically,  such that these properties are met. In Section~\ref{sec_buckets_simulations}, we illustrate the good performance of the seMCD-methodology by means of simulations and real data analyses, followed by some discussions. The appendix contains  more detailed information about important depth functions suitable for the seMCD-methodology.

\section{Depth functions and their typical applications}
\label{section_2_depth}

The methodology in this paper is designed for situations where the quantity of interest is based on depth functions and can be written as 
\begin{align}\label{E} 
\E(H(\bm \xi)).
 \end{align}
Here, $H$ is a known function, and $\bm \xi$ is a random element. Clearly, when $H$ is computable and drawing iid samples from the distribution of $\bm{\xi}$ is possible, the quantity $\E(H(\bm{\xi}))$ can be approximated by Monte Carlo methods. 

We refer to the depth functions, whose values can be represented as in \eqref{E}, as E-depth functions and give some explicit examples in Section \ref{Exdepth}, including both theoretic and empirical depth calculation. Section \ref{type_b} is devoted to other depth functions that also fit the framework. Meanwhile, in Section \ref{appl} we provide some examples of applications of statistical depth that fit into the seMCD-framework.  

\subsection{E-depth functions}\label{Exdepth}
Given a space $\mathbb{S}$, that is generally multivariate or functional, the depth of an element $\bz\in\mathbb{S}$ is calculated with respect to a distribution $P$ on $\mathbb{S}$. This depth value is denoted as $D(\bz;P)$. When $D(\bz;P)$ or a related depth-based quantity has a representation as in \eqref{E}, the distribution of $\bm{\xi}$ is often associated with $P$. Moreover, $H$ typically depends on $\bz$, as well as the depth function used (for instance, simplicial or band depth -- see  \eqref{hs} or \eqref{band} in Appendix~\ref{subsec_multivariate_depths_A}) and the application at hand (such as depth calculation or classification, see Section~\ref{appl}).

In practice, we want to estimate  the depth values using the realisations $\bz_1,\ldots,\bz_n$ of some random elements $\bZ_1,\ldots,\bZ_n\overset{\textrm{iid}}{\sim} P$. The standard estimator corresponding to \eqref{E} is of the form 
\begin{align}\label{eq_emp_depth_cond_expRV}
\E(H_{\bZ_1,\ldots,\bZ_n}(\bm \zeta)\,|\,\bZ_1,\ldots,\bZ_n),
\end{align}
with some function $H_{\bZ_1,\ldots,\bZ_n}$ related to $H$, and a suitable random element $\bm \zeta$; see \eqref{typeAdepth} and \eqref{Deallgemein} below for an example. The random element $\bm\zeta$ encodes the Monte Carlo draw, thus having a known distribution and being independent of $\bZ_1,\ldots,\bZ_n$. Hence, the corresponding estimate is given by 
\begin{align}\label{eq_emp_depth_cond_exp}
     \E(H_{\bZ_1,\ldots,\bZ_n}(\bm \zeta)\,|\,\bZ_1 =\bz_1,\ldots,\bZ_n=\bz_n) =\E(H_{\bz_1,\ldots,\bz_n}(\bm \zeta)).
\end{align}
As long as $H_{\bz_1,\ldots,\bz_n}$ can be computed and $\bm \zeta$ can be sampled from, it is possible to approximate\linebreak $\E(H_{\bz_1,\ldots,\bz_n}(\bm \zeta))$ by Monte Carlo methods. From a computational point of view, given a specific data realisation, the quantity of interest has again the form \eqref{E} with $H:=H_{\bz_1,\ldots,\bz_n}$ and $\bm{\xi}:=\bm \zeta$. Typically, the $H_{\bz_1,\ldots,\bz_n}$ related to the empirical depth is different from the $H$  corresponding to the theoretical depth and also the domain of the functions can differ entirely. 

The class of E-depth functions includes Type A depth functions \cite[Section 2.3.1]{ZuoSerfling} but is not limited to these.

\paragraph{Type A depth functions.} The theoretical depth of $\bz\in\mathbb{S}$ with respect to a distribution $P$ on $\mathbb{S}$ in the class of Type A depth functions is characterised by 
\begin{align}\label{DA}
    D_A(\bz;P) = \E \big(G(\bz;[\bm{\eta}_1,\ldots,\bm{\eta}_r]) \big),
\end{align} 
    where $\bm{\eta}_1, \ldots,\bm{\eta}_r \overset{\textrm{iid}}{\sim} P $  for some fixed natural number $r$, with $r$ and $G$ depending on the depth and possibly the dimension of the space. This naturally fits in \eqref{E} by setting $\bm{\xi}:=[\bm{\eta}_1,\ldots,\bm{\eta}_r]$ and $H(\bm{\xi}):=G(\bz; \bm{\xi})$. 

For applications, users are interested in estimators of theoretical depth values. These estimators are based on a random sample $\bZ_1,\ldots,\bZ_n$, $n\ge r$, because the underlying $P$ is generally unknown. 
Usually, $G$ is symmetric,  i.e.\ invariant under permutations of the $r$ arguments in $[\bm{\eta}_1,\ldots,\bm{\eta}_r]$. In this case, a popular estimator of $ D(\bz;P)$ with respect to $\bZ_1,\ldots,\bZ_n$ is given by the corresponding U-statistic
    \begin{align}\label{typeAdepth}
        D_A(\bz;\bZ_1, \ldots,\bZ_n) = {n \choose r}^{-1} \sum_{1 \leq i_1<...<i_r \leq n} 
        G(\bz; [\bZ_{i_1}, \ldots,\bZ_{i_r}]).
    \end{align} 
 Although the sum consists of only a finite number of terms, the total number of summands is given by $\binom{n}{r}$, i.e.\ of order $n^r$, which can be excessively large for efficient computation, even for moderate~$r$. This issue can be solved by using the seMCD framework instead, because the above sum can be written as an expectation. Indeed, some simple calculations reveal that
 \begin{align}
    \label{Deallgemein}
	    D_A(\bz;\bZ_1,\ldots,\bZ_n)&=
	    \E\left(\left. G(\bz;[ \bZ_{\zeta_1},\ldots,\bZ_{\zeta_r}])\right| \bZ_1,\ldots,\bZ_n\right),
    \end{align}
where  $(\zeta_1,\ldots,\zeta_r)$ is a random $r$-combination from the set $\{1,\ldots,n\}$ independent of $\bZ_1,\ldots,\bZ_n$ with $P((\zeta_1,\ldots,\zeta_r)=\bm{c})={n \choose r}^{-1}$ for any $\bm{c}\in \{(c_1,\ldots,c_r): 1\le c_1<c_2<\ldots<c_r\le n\}$. Consequently, this empirical depth fits the framework in \eqref{eq_emp_depth_cond_expRV}, and \eqref{eq_emp_depth_cond_exp}, with $\bm \zeta:=(\zeta_1,\ldots,\zeta_r)$ and $H_{\bZ_1,\ldots,\bZ_n}(\bm \zeta):=G(\bz;[ \bZ_{\zeta_1},\ldots,\bZ_{\zeta_r}])$. 

Among Type A depth functions, an important subclass is what we call Type A depth functions with indicator kernel. 
Examples of multivariate depth functions in this subclass are simplicial depth, spherical depth \citep{elmore2006}, lens depth \citep{LiuModarres2011} and $\beta$-skeleton depth \citep{yang2018beta}. While the Type A depth class was originally conceived in \cite{ZuoSerfling} for multivariate spaces, some functional depths also belong to it. Instances are the band depth based on two curves and the norm-based statistical depth family \citep{Zhao2024}, both in the subclass with indicator kernel. Other well-known instances are the h-depth \citep{Cuevas2007Projection} and the modified band depth \citep{LR2009} when they are based on two curves. 
More information on how these depth functions fit the framework is provided for multivariate spaces in Appendix~\ref{subsec_multivariate_depths_A}  and for functional spaces in Appendix~\ref{subsec_functional_depth_A}. 

\paragraph{Integrated depths.}
We refer to integrated depth functions   as those for which the theoretical depth of $\bz\in\mathbb{S}$ with respect to a distribution $P$ on $\mathbb{S}$  is given by 
\begin{align}\label{DI}
    D_I(\bz;P) = 
    \E \big(I(\bz;\bm{\eta}) \big),
\end{align} 
with $I$  depending on the depth and $P$  while $\bm\eta$ does not relate to $P$. The expectation in \eqref{DI} is with respect to $\bm{\eta}$.  A very well-known example is the integrated dual depth \citep{CuevasFraimanIntegrated2009}, with the integrated rank-weighted (IRW) depth \citep{Ramsay2019,Staerman2021} as a special case. For further details see Appendix~\ref{subsec_IRW_depth}.

\subsection{Monotonic E-depth functions}\label{type_b}
The seMCD-methodology is also suitable for calculating depth regions and detecting outliers  (see Section \ref{appl}) in combination with what we call monotonic E-depth functions. There, the quantity of interest is given by $F(\E(H(\bm\xi)))$, for a known strictly monotonic function $F$. In particular, this covers all Type B depth functions, such as the simplicial volume depth and the $L^p$ depth family \citep{ZuoSerfling}. Other depth functions that fit this framework are the spatial depth \citep{SerflingSpatial2002,SpatialDepth} and the  kernelized spatial depth \citep{Sguera2014}. 

In classification based on monotonic E-depth functions, some limitations occur for maximal-depth classification in combination with so-called overlapping buckets. The latter are introduced in Section~\ref{sec_new_overlapp}, and the limitations are discussed in Section~\ref{sec_transformation}.

\subsection{Typical applications of depth functions}\label{appl}
In the previous two sections, we have detailed several depth functions that fall within the framework of this paper. In this section, we will highlight some typical applications of depth functions, that can be treated with the seMCD-methodology, if combined with the above E-depth or monotonic E-depth functions. 

Recall that the seMCD-methodology outputs an interval of values, the seMCD-bucket, that contains the unknown quantity of interest $\E(H(\bm\xi))$, or $F(\E(H(\bm\xi)))$, with a user-chosen probability. We will highlight why obtaining such buckets is of most interest in important applications of depths functions.

\subsubsection{Depth calculation and level sets}\label{section_depth_level_sets}
In many situations, it is not necessarily the precise depth value that is of interest, but whether or not it falls within a certain range. In this context, the buckets correspond to user-specified level regions (possibly overlapping). For example, a choice of overlapping seMCD-buckets for a depth function with values in $[0,1/2]$ could be $\{[0,1/10),(1/20,3/20),(1/10,2/10),(3/20,5/20),\ldots, (7/20,9/20),(4/10,5/10]\}$.
The advantage of such an approach is that it typically requires fewer samples than standard Monte Carlo studies. In addition, the output seMCD-bucket contains the quantity of interest with a certain statistical guarantee.

\subsubsection{Outlier or anomaly detection}\label{section_outlier}
A well-known application of depth functions is outlier detection \citep[and the papers discussing it]{Hubert2015}. Data points with a depth smaller than some prespecified value $c_o$ are considered potential outliers. Of course, low depth values do not necessarily need to belong to outliers nor do outliers necessarily need to have small depth values. This is particularly the case for functional data \citep{Arribas}, where different depths will focus on different properties of the functions and thus potentially identify different observations as possible outliers. In this sense, depth-based outlier detection should be considered more as a method for detecting anomalies \citep{mozharovskyi2024anomalydetectionusingdata}, which may or may not occur naturally, but which require special attention by the statistician or data analyst. The use of depths for anomaly detection naturally fits into the framework of pre-specified buckets because the exact depth value is not of so much interest as whether or not it falls in a low-depth region with a depth smaller than $c_o$. In this case, natural seMCD-buckets are $[0,c_o), (c_o,\infty)$,  in addition to $(c_o-\epsilon,c_o+\epsilon)$ for some $0<\epsilon<c_o$ in the overlapping case.

\subsubsection{Depth-based classification}\label{subsec_classification}
Another natural application of depth buckets is classification. Consider a point $\bz$ and  $\ell$ classes, with $\ell \in \mathbb N$ fixed, with underlying empirical probability measures $\widehat P_{r_-1}, \ldots,\widehat P_{r_{\ell}}$. One approach is to determine a bucket for each empirical depth value $D(\bz,\widehat P_{r_1}), \ldots,D(\bz,\widehat P_{r_{\ell}})$ and then classify $\bz$ into the class with the highest bucket value. In order to avoid simultaneous assignment to multiple classes, it is recommendable to make use of a sufficiently fine grid of overlapping seMCD-buckets. 

For binary classification, based on the realisations of $\ell=2$  classes $\bx_1,\ldots,\bx_{r_x};\by_1,\ldots,\by_{r_y}$, with (empirical) E-depth functions \eqref{eq_emp_depth_cond_exp} it is possible to reduce the computational effort by directly approximating the difference in the two empirical depth values 
\begin{align}\label{eq_classification_difference}
    \E[H_{\bx_1,\ldots,\bx_{r_x}}(\bm \zeta)-H_{\by_1,\ldots,\by_{r_y}}(\bm \zeta)]
\end{align}
via Monte Carlo studies. Choosing $\bm{\xi}:=\bm \zeta$ and $H:= H_{\bx_1,\ldots,\bx_{r_x}}-H_{\by_1,\ldots,\by_{r_y}}$, this is again of the desired form $\E(H(\bm{\xi}))$. The maximum depth classifier will classify a value $\bz$ as belonging to the $\bx$-sample, if this expectation is positive, and to the $\by$-sample otherwise. Consequently, natural seMCD-buckets for this problem are $(-\infty,0)$, $(0,\infty)$. Additionally, a third bucket $(-\epsilon,\epsilon)$ for some value $\epsilon>0$ can be added in the overlapping case. For the restrictions that apply concerning the use of monotonic E-depth functions and overlapping buckets, see Section~\ref{sec_transformation}.

\section{Sequentially implemented Monte Carlo depth methodology} \label{sec_main_seMCD}
In this section, we first revisit standard Monte Carlo techniques that are applicable to approximate E-depths. We will refer to them as plain-vanilla Monte Carlo. In Section \ref{SeMCD} we  introduce our proposed \textbf{se}quential \textbf{M}onte \textbf{C}arlo \textbf{D}epth (seMCD) methodology as an alternative with statistical guarantees in the context of applications for E-depths. In Section~\ref{sec_transformation} we explain how to adapt the procedure for monotonic E-depths.
    
\subsection{Plain-vanilla Monte Carlo computation}\label{sec_PlainVanilla}
As explained in Section~\ref{Exdepth}, all E-depths (whether theoretical or empirical) are of the form $\E(H(\bm{\xi}))$ for a known and computable function $H$ and a random element $\bm\xi$ that one can easily take samples of. If the exact calculation of this expected value is not possible, such E-depths can easily be approximated by Monte Carlo simulations of the type
\begin{align}\label{eq_plainVanillaMC}
	\frac 1 N \sum_{j=1}^N H(\bm{\xi}_j),
\end{align}
where $\bm{\xi}_1,\ldots,\bm{\xi}_N$ are iid\ copies of $\bm{\xi}$. Not all depth functions can be approximated by sums. A prominent counterexample is the Tukey depth, as its empirical version involves taking an infimum over a function of samples but not a sum. Thus, the methodology in this paper is not applicable to such depths.   

The plain-vanilla Monte Carlo computation  has two associated problems:
\begin{itemize}
	\item The number $N$ of Monte Carlo repetitions usually has to be chosen quite large in order for the law of large numbers to give a precise enough approximation. This makes the plain-vanilla approach computationally quite costly. For example, for the IRW depth \citet[Example 4]{Ramsay2019} take $N=10^5$ and \cite{Staerman2021} recommend $N=100 \cdot d$. 
	\item Often, one is interested in making a subsequent decision based on the quantity of interest. For example, in anomaly detection, the question is whether the true depth value is smaller than a given threshold. In such a case,  there is no statistical guarantee that the plain-vanilla approximation is indeed on the same side of the decision region as the true value. This is particularly problematic if the true value is close to the decision threshold (or even exactly at it), because then even a small difference between the sample mean in \eqref{eq_plainVanillaMC} and the corresponding expected value can lead to a different decision.
\end{itemize}
The seMCD computation  gives such statistical guarantees and typically largely reduces the number of necessary Monte Carlo samples and, consequently, the computation time.  

\subsection{SeMCD computation with statistical guarantees}\label{SeMCD}
We will now first explain the seMCD methodology in combination with E-depths, where the quantity of interest is given by $h:=\E(H(\bm{\xi}))$.

The seMCD-algorithm takes the following user inputs: a tolerance $\alpha>0$ as well as a set of buckets, to which we refer as seMCD-buckets, typically covering any potential value of $h$. 
The seMCD-algorithm will output one of these buckets with the guarantee that it contains the true value with a probability (either exactly or asymptotically, depending on the chosen depth and methodology) greater than or equal to $1-\alpha$. Unlike confidence intervals, the potential output intervals are determined in advance by fixing the seMCD-buckets.

The algorithm takes Monte Carlo samples one by one until the desired guarantee is reached, i.e.\ until there is sufficient statistical evidence that one of the given seMCD-buckets contains the true value $h$.  The number of Monte Carlo samples required  is random, because it depends on the actual samples taken in any given run of the algorithm. It also depends on the given tolerance, the given seMCD-buckets and the typically unknown distribution of $H(\bm{\xi})$, including the unknown true value $h=\E(H(\bm{\xi}))$.  The smaller the $\alpha$ and the closer the true value is to the boundaries of the buckets in which it is contained, the more Monte Carlo samples are required to provide the statistical guarantee. In particular, smaller buckets typically require more Monte Carlo samples. However, in many cases, one can return the seMCD-bucket with far fewer Monte Carlo samples than are commonly used in plain-vanilla computations.

The seMCD-algorithm  borrows ideas from sequential testing and is  based on the following seMCD-process 
\begin{align}\label{eq_SM}
	S_N:=\sum_{j=1}^NH(\bm{\xi}_j), \quad N \in \mathbb N,
\end{align}
where $\{\bm{\xi}_j: j\ge 1\}$ are independent copies of $\bm{\xi}$. In addition, the algorithm is based on suitable boundary sequences chosen in such a way that the desired statistical guarantees hold. We discuss the construction of these boundary sequences in Section~\ref{sec_boundary}. Section~\ref{sec_gandy_approach} treats the case of $H(\bm{\xi})$ having a Bernoulli distribution with unknown success probability $h=\E(H(\bm{\xi}))$, as happens for all Type A depths with indicator kernel, including the simplicial depth. In this case, the statistical guarantees are exact.  Section~\ref{sec_sequential_construction} provides sequences in a non-parametric setting that give only asymptotic guarantees (in an appropriate sense, to be discussed there).  On the plus side, these bounds are suitable for all applications and E-depths of this paper. We also give an example where such non-parametric methods are unavoidable.

In the remainder of this subsection, we discuss how the algorithm works given such boundary sequences.
Subsections \ref{sec_two_buckets} and \ref{sec_new_nonoverlapp} are dedicated to the case in which we have non-overlapping seMCD-buckets. 
These two subsections will serve as motivation for our preferred proposed methodology that consists of overlapping buckets, and is developed in Subsection \ref{sec_new_overlapp}. In Section~\ref{sec_stat_guar} we give some insight into the statistical guarantees associated with the algorithms, before introducing a greedy algorithm in Section~\ref{sec_another_look_algorithm}.

\subsubsection{Two seMCD-buckets case}\label{sec_two_buckets}
Let us first explain the seMCD-algorithm in the simplest situation, where there are only two seMCD-buckets, $(-\infty,h_1)$ and $(h_1,\infty)$  for $h=\E (H(\bm{\xi}))$.  We will later elaborate on the openness at the split point $h_1$. In many cases, the set of potential values of $h$ is bounded. For example, many depths only take values between $0$ and $1$. If this is the case, the above seMCD-buckets are equivalent to $[0,h_1)$ and $(h_1,1]$. Two applications, for which two seMCD-buckets may be sufficient, are: (1) anomaly detection, with, for example, $h_1=0.05;$ and (2) binary maximum depth classification, with $h_1=0$ corresponding to equal depth within both samples.

The seMCD-algorithm in the two seMCD-buckets case is based on two boundary sequences, an upper boundary sequence $\{U_N^{(h_1)}:N\in\mathbb{N}\}$ and a lower boundary sequence $\{L_N^{(h_1)}:N\in\mathbb{N}\}$ with $L_N^{(h_1)}\le U_N^{(h_1)}$ for all $N\in\mathbb{N}$ that fulfil,
\begin{align}
	&\text{if }h=h_1,\quad \P_{h}\left( L_N^{(h_1)}<S_N<U_N^{(h_1)} \text { for all }N\in\mathbb{N}\right)\ge 1-\alpha, \tag{I} \label{I}\\
	&\text{if }h<h_1,\quad\P_h\left(S_N<U_N^{(h_1)} \text{ for all  }N\in\mathbb{N}\right)\ge 1-\alpha,\tag{II}\label{II}\\
	&\text{if }h>h_1,\quad\P_h\left(L_N^{(h_1)}<S_N\text{ for all  }N\in\mathbb{N}\right)\ge 1-\alpha,\tag{III}\label{III}
\end{align}
where $\alpha \in (0,1)$ is the chosen tolerance. In the non-parametric situation, the precise inequalities in \eqref{I}-\eqref{III} will be replaced by appropriate asymptotic statements,  see Section~\ref{sec_sequential_construction} for the details.

Usually, the assertions in \eqref{II} and \eqref{III} follow from those in \eqref{I} by monotonicity of the probabilities with respect to $h$.  For example, if the sequences associated with $\mathfrak{h}<h_1$ satisfy the monotonicity condition $U_N^{(\mathfrak h)}\le U_N^{(h_1)}$ for all $N\in\mathbb{N}$, then clearly \eqref{II} for $h=\mathfrak{h}$ follows from a condition as in \eqref{I}  with $h_1$ replaced by $\mathfrak h$. Similarly, if the sequences associated with $\mathcal{h}>h_1$ satisfy $L_N^{(h_1)}\le L_N^{(\mathcal h)}$,  then  \eqref{III} follows from \eqref{I} with $h_1$ replaced by $\mathcal h$.

The algorithm is described as follows:

\vspace{2mm}
\begin{mdframed}
\textbf{\underline{seMCD-algorithm for two seMCD-buckets}}\\[2mm]
Continue taking samples until, for the first time, \begin{itemize}
	\item either $S_N\ge U_N^{(h_1)}$: In this case, the algorithm outputs the seMCD-bucket $(h_1,\infty)$.
	\item or $S_N\le L_N^{(h_1)}$: In this case, the algorithm outputs the seMCD-bucket $(-\infty,h_1)$.
\end{itemize}
\end{mdframed}
\vspace{2mm}

Clearly, conditions \eqref{II} and \eqref{III} guarantee that the probability of stopping in a wrong seMCD-bucket is bounded by the chosen tolerance $\alpha$ for any $h\neq h_1$. Lemma~\ref{lem_new_stop}, in Section \ref{sec_stat_guar},  gives a sufficient condition under which the above seMCD-algorithm stops with probability one for $h\neq h_1$. For the split point $h=h_1$, i.e.\ directly at the decision boundary, \eqref{I} clearly shows that  the algorithm will never stop with a probability of at least $1-\alpha$. This is the reason why we  have defined the seMCD-buckets as open buckets. In Section~\ref{sec_new_overlapp} we will solve the non-stopping problem at $h=h_1$ by defining overlapping seMCD-buckets, following \cite{Gandy2020b}.

\begin{figure}
        \centering
        \begin{subfigure}[h]{0.325\columnwidth}
        \includegraphics[width=\columnwidth, trim={4.0cm 10cm 4.5cm 10cm},clip]{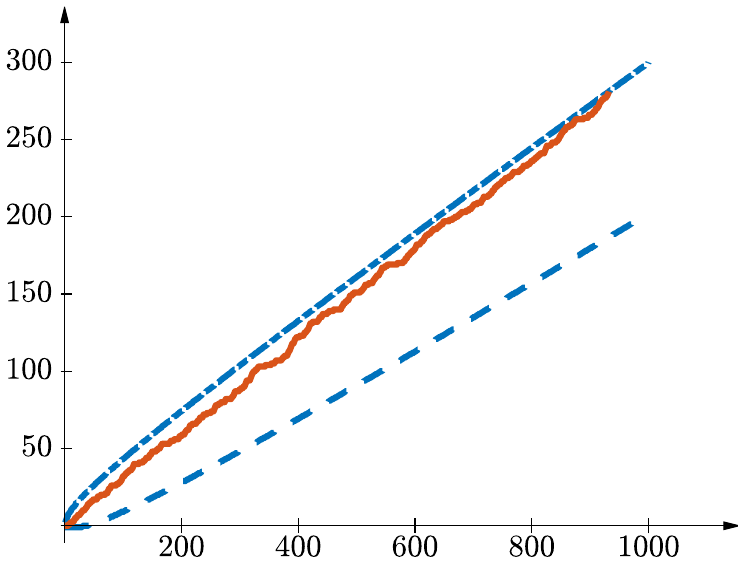}
        \vspace*{-4ex}
	\caption{\centering Decision for $h>0.25$ after\linebreak $N=932$ sampling steps\linebreak with a true $h$ of 0.3.}
    \end{subfigure}
    \begin{subfigure}[h]{0.325\columnwidth}
        \includegraphics[width=\columnwidth, trim={4.0cm 10cm 4.5cm 10cm},clip]{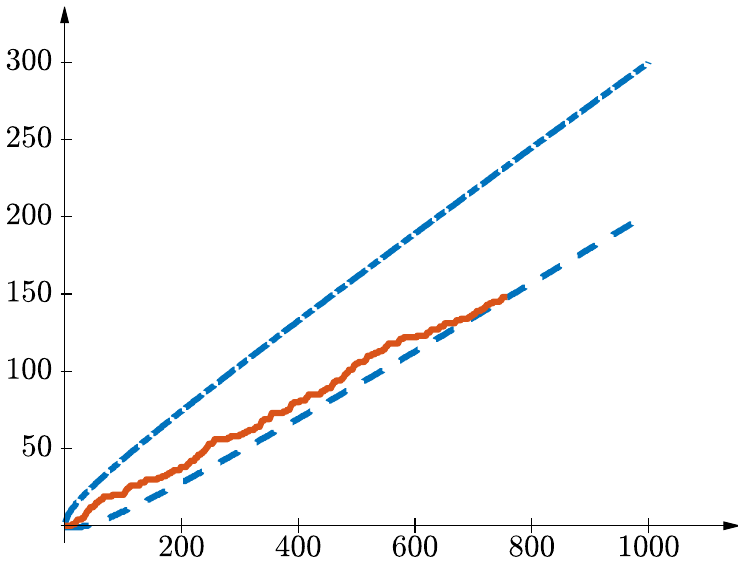}
        \vspace*{-4ex}
        \caption{\centering Decision for $h<0.25$ after $N=758$ sampling steps\linebreak with a true $h$ of 0.2.}
    \end{subfigure}
    \begin{subfigure}[h]{0.325\columnwidth}
        \includegraphics[width=\columnwidth, trim={4.0cm 10cm 4.5cm 10cm},clip]{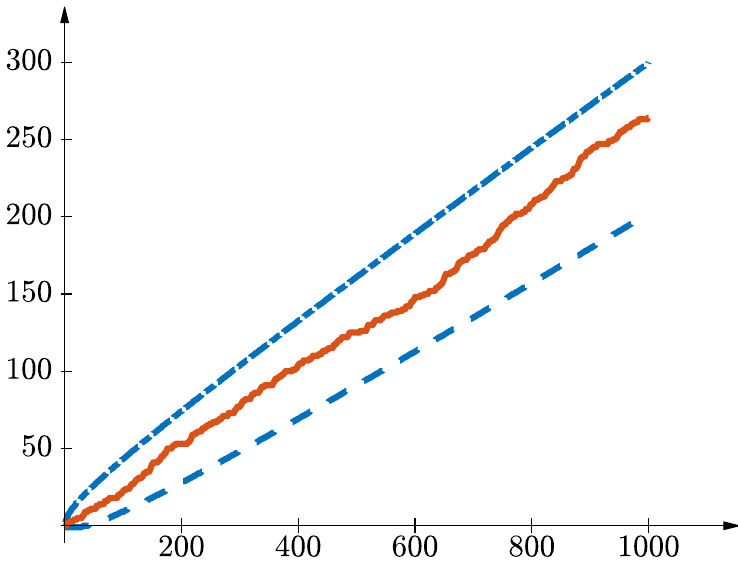}
        \vspace*{-4ex}
        \caption{\centering No decision after $N=1000$ sampling steps with a true $h$ equaling the split point $0.25$.}
    \end{subfigure}
    \caption{Upper boundary sequence $\{U_N^{(0.25)}:N \in\mathbb{N}\}$ (blue dashed-dotted) and lower boundary sequence $\{L_N^{(0.25)}:N \in\mathbb{N}\}$ (blue dashed) for tolerance $\alpha = 0.01$ and  $H(\bm{\xi})\sim\textrm{Bernoulli}(0.25)$. Three seMCD-sample paths $\{S_N: N \in\mathbb{N}\}$ representing different outcomes for stopping the algorithm: orange sequences.}
    \label{fig_stopping_two_outcomes}
\end{figure}

Figure \ref{fig_stopping_two_outcomes} gives an illustration of the seMCD-algorithm using Bernoulli$(h)$-distributed random variables $H(\bm{\xi}_j)\mbox{, } j\ge 1$, with $h=\E (H(\bm{\xi}_1)) \in \{0.2,0.25,0.3\}$ and the seMCD-buckets $(0,0.25),(0.25,1)$. This reflects the situation of Type A depths with indicator kernel including e.g.\ the simplicial depth.  In (a) the true value of $h$ equals $0.3$ and the algorithm returns the bucket $(0.25,1)$, while in (b) $h=0.2$ holds and a decision for the bucket $(0,0.25)$ is made. In (c) $h = 0.25$, and the algorithm returns no decision within the first 1000 steps. The boundary sequences were constructed as in Section \ref{sec_gandy_approach} using $\alpha = 0.01$ as tolerance and the default spending sequence

\begin{align}\label{eq_default_spending}
    \alpha_N = \alpha \cdot N/(1000+N), \quad N \in \mathbb N.   
\end{align}

In practice, the seMCD-algorithm will stop more quickly the farther the true $h$ is from the split point $h_1$. The number of samples taken before the algorithm stops is typically much smaller than the usual plain-vanilla Monte Carlo sample sizes.

\subsubsection{Non-overlapping multiple seMCD-buckets}\label{sec_new_nonoverlapp}
While returning only two possible seMCD-buckets is sufficient for some applications, this may\linebreak be too coarse a decision for others. Similarly to using depth regions, one may instead be\linebreak interested in distinguishing finitely many seMCD-buckets $(-\infty,h_1), (h_1,h_2), \ldots,(h_{k},\infty)$ with\linebreak $-\infty=h_0<h_1<...<h_{k}<h_{k+1}=\infty$, not necessarily equidistant, with split points $h_1,\ldots,h_k$ in the range of possible values of $h$. As before, if e.g.\ the range of possible values is given by $0\le h\le 1$, this is equivalent to setting $h_0=0$ and $h_{k+1}=1$. 

The seMCD-algorithm can be adapted to this non-overlapping multi-bucket situation as follows: For each split point  $h_j$, $j=1,\ldots,k$, we need lower and upper boundary sequences $\{L_N^{(h_j)}:N \in\mathbb{N}\}$ and $\{U_N^{(h_j)}:N \in\mathbb{N}\}$  as in Section~\ref{sec_two_buckets}, satisfying \eqref{I} -- \eqref{III} for the tolerance $\alpha$. We additionally require the following monotonicity condition to hold on the boundary sequences for arbitrary values:

\begin{Condition}\label{ass_Monotonnicity_boundary_sequences} 
	For any $h_{-}<h_+$ in the range of possible values of $h$ it holds  
	\begin{align} \label{eq_Monotonnicity_boundary_sequences}      
        U_N^{(h_-)} \leq U_N^{(h_+)} \qquad\text{and}\qquad L_N^{(h_-)} \leq L_N^{(h_+)}. 
    \end{align}
\end{Condition}

In Section~\ref{sec_stat_guar} we show that the statistical guarantees from the  two-bucket situation carry over to the following algorithm for multiple non-overlapping buckets. 

\vspace{2mm}
\begin{mdframed}
\textbf{\underline{seMCD-algorithm for multiple non-overlapping seMCD-buckets}}\\[2mm]
Define $U_N^{(h_0)}=-\infty$ and $L_N^{(h_{k+1})}=\infty$  
for all $N$.\\[1mm]
Continue taking samples until, for the first time $\tau$, it holds $\mathfrak{r}_\tau-\mathfrak{l}_\tau=1$, where, for arbitrary $N$,
\begin{align*}
&\mathfrak{l}_N := \max\{0\le j\le k: S_N\ge U_N^{(h_j)}\},\qquad \mathfrak{r}_N := \min\{\mathfrak{l}_N< j \le k+1: S_N\le L_N^{(h_j)}\}.
\end{align*}
Output the seMCD-bucket $(h_{\mathfrak{l}_{\tau}},h_{\mathfrak{r}_{\tau}})$.
\end{mdframed}

Figure \ref{fig_nonstopping_regions} (a) gives an illustration: As soon as the seMCD-process $S_N$ enters the white area above the green line, it stops and returns the bucket $(0.4,1]$. If it first enters the white area between the blue and green lines, then it stops and returns $(0.25,0.4)$. Between the orange and the blue lines it returns $(0.1,0.25)$ and below the orange line it returns $[0,0.1)$.
 As a general rule, if the lengths of the correct seMCD-buckets are smaller, more sampling steps can be expected before a decision is reached. 

As in the case of only two  seMCD-buckets, for $h$ equal to a split point, i.e.\ $h \in \{h_1, \ldots,h_{k}\}$, the seMCD-algorithm will never stop with a probability greater than $1-\alpha$. This corresponds to the situation, where the seMCD-process remains within the gray area in Figure \ref{fig_nonstopping_regions} (a), which happens with a probability of at least $1-\alpha$ for $h=h_j$, $j=1,\ldots,k$, by \eqref{I} for that $h$. 

\subsubsection{Overlapping seMCD-buckets}\label{sec_new_overlapp}
In order to overcome the non-stopping problem discussed above, we will follow the ideas in \cite{Gandy2020b} and use overlapping buckets instead. Consider now, $-\infty=h_0<h_1<...<h_k<h_{k+1}=\infty$, not necessarily equidistant, in the range of possible values for $h$ and consider the following overlapping seMCD-buckets:
\begin{equation}\label{eq_ordering_buckets}
\vspace{-1mm}
\begin{tikzpicture}[baseline=(current  bounding  box.center)]
\draw[thick,->](-0.5,0) -- (9.5,0);

\foreach \x in {1, 2,3, 4,5,6,7,8} {
        \draw (\x, 0.1) -- (\x, -0.1); 
        
    }

\foreach \x in {1,  3,7} {
        \draw (\x, 0.1) -- (\x, -0.1); 
        \node[above] at (\x, -0.7) {,};
        \node[above] at (\x+1, 0.1) {,};
        
    }
    \node[above] at (5,-0.7) {$\ldots$};
     \node[above] at (6,0.1) {$\ldots$};

    \node[above] at (0,-0.75) {$($};
    \node[above] at (0.5,-0.68) {$-\infty$};
    \node[above] at (1.5,-0.7) {$h_2$};
    \node[above] at (2.5,-0.7) {$h_2$};
        \node[above] at (2,-0.75) {$)($};
        \node[above] at (3.5,-0.7) {$h_4$};
        \node[above] at (4,-0.75) {$)$};
  \node[above] at (6,-0.75) {$($};
    \node[above] at (6.5,-0.75) {$h_{k-2}$};
    \node[above] at (7.5,-0.75) {$h_k$};
    \node[above] at (8,-0.75) {$)$};

      \node[above] at (1,0.1) {$($};
    \node[above] at (1.5,0.1) {$h_1$};
    \node[above] at (2.5,0.1) {$h_3$};
    \node[above] at (3.5,0.1) {$h_3$};
        \node[above] at (3,0.1) {$)($};
        \node[above] at (4.5,0.1) {$h_5$};
        \node[above] at (5,0.1) {$)$};  
         \node[above] at (7,0.1) {$($};
    \node[above] at (7.5,0.1) {$h_{k-1}$};
    \node[above] at (8.5,0.2) {$\infty$};
    \node[above] at (9,0.1) {$)$};

\end{tikzpicture}
\end{equation}

Unlike before, all split points $h_j$, $j=1,\ldots,k$, are now included in the interior of another seMCD-bucket. All  points between $h_1$ and $h_k$, that are not split points, are included in the interior of two different seMCD-buckets. In all cases, the algorithm will output one of the correct seMCD-buckets, or in exceptional cases their intersection, with the same guarantees as before. 
 
For applications as in Section~\ref{section_depth_level_sets}  with $0\le h\le 1$, we recommend a gapless covering of the parameter domain with seMCD-buckets of the same length, see \eqref{eq_ordering_buckets}, such that $h_j=(h_{j+1}-h_{j-1})/2$ with $h_0=0, h_1=1$, but other configurations are, of course, also possible. On the other hand, for anomaly detection as in Section~\ref{section_outlier} or classification as in Section~\ref{subsec_classification}, it makes sense to have a small third bucket covering the original decision boundary.

Unlike in the situation of non-overlapping seMCD-buckets, the algorithm will now stop under weak conditions for any $h$ including the split points, and usually much earlier than for non-overlapping seMCD-buckets or the usual plain-vanilla Monte Carlo sample sizes. Figure \ref{fig_nonstopping_regions} (b) illustrates this fact with an example: The sampling is guaranteed to stop if the seMCD-process $S_N$ reaches the white area. In this example, the gray area ends after the $M=494$ steps, so the seMCD-algorithm will have stopped for all realizations of the seMCD-process at that point.

\begin{figure}
        \centering
        \begin{subfigure}[h]{0.48\columnwidth}
        \includegraphics[width=\columnwidth, trim={4.0cm 10cm 4.5cm 10cm},clip]{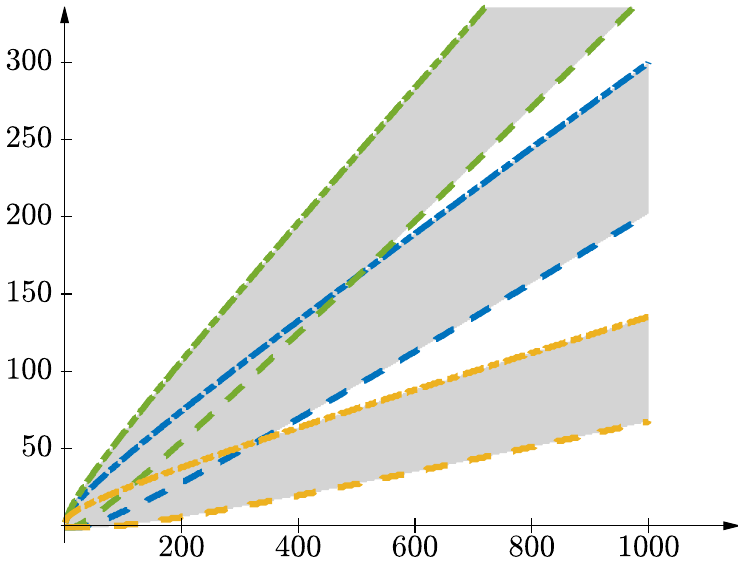}
        \vspace*{-4ex}
        \caption{\centering non-overlapping seMCD-buckets $[0,{\color{plotorange} 0.1})$, $({\color{plotorange} 0.1},{\color{plotblue} 0.25})$, $({\color{plotblue} 0.25},{\color{plotgreen} 0.4})$, $({\color{plotgreen} 0.4},1]$.}
        \captionsetup{justification=centering}
    \end{subfigure}
    \begin{subfigure}[h]{0.48\columnwidth}
        \includegraphics[width=\columnwidth, trim={4.0cm 10cm 4.5cm 10cm},clip]{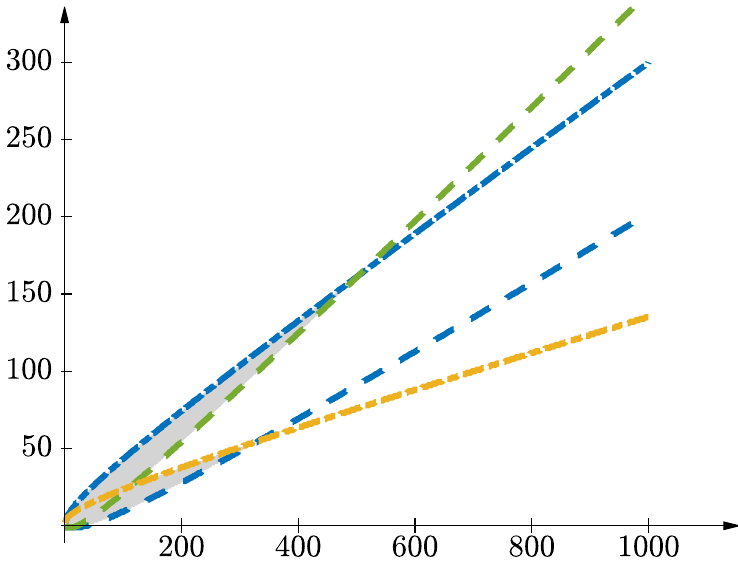}
        \vspace*{-4ex}
        \caption{\centering Overlapping seMCD-buckets $[0,{\color{plotblue} 0.25})$,\linebreak $({\color{plotorange} 0.1},{\color{plotgreen} 0.4})$, $({\color{plotblue} 0.25},1]$.}
    \end{subfigure}
    \caption{Stopping regions (white) for non-overlapping seMCD-buckets (left panel) and overlapping seMCD-buckets (right panel) for the Type A depth functions with indicator kernel. As soon as the seMCD-process $\{S_N\}$ is located in the white area, the algorithm will stop. Dashed lines are lower boundaries: Decision for $p<{\color{plotorange} 0.1}$ (orange), $p<{\color{plotblue} 0.25}$ (blue), $p<{\color{plotgreen} 0.4}$ (green). Dash-dotted lines are upper boundaries: Decision for $p>{\color{plotorange} 0.1}$ (orange), $p>{\color{plotblue} 0.25}$ (blue), $p>{\color{plotgreen} 0.4}$ (green).}
    \label{fig_nonstopping_regions}
    \end{figure}

The algorithm is the same as before, with the only difference that it stops as soon as $\mathfrak{r}_N-\mathfrak{l}_N\le 2$. If $\mathfrak{r}_N-\mathfrak{l}_N= 2$, then one of the original seMCD-buckets is output. If $\mathfrak{r}-\mathfrak{l}= 1$, then the non-empty intersection of two of these buckets is output. The latter can only happen if the additional sample changes the partial sum process so dramatically that it crosses more than one of the boundary sequences with the addition of this new sample, which will rarely happen in practice. We phrase the algorithm for completeness.

\begin{nobreaks}
\vspace{2mm}
\begin{mdframed}
\textbf{\underline{seMCD-algorithm for multiple overlapping seMCD-buckets  as in \eqref{eq_ordering_buckets}}}\\[2mm]
Define $U_N^{(h_0)}=-\infty$ and $L_N^{(h_{k+1})}=\infty$  
for all $N$.\\[1mm]
Continue taking samples until, for the first time $\tau$, it holds $\mathfrak{r}_\tau-\mathfrak{l}_\tau\le 2$, where, for arbitrary $N$,
\begin{align*}
&\mathfrak{l}_N := \max\{0\le j\le k: S_N\ge U_N^{(h_j)}\},\qquad \mathfrak{r}_N := \min\{\mathfrak{l}_N< j \le k+1: S_N\le L_N^{(h_j)}\}.
\end{align*}
Output the seMCD-bucket $(h_{\mathfrak{l}_{\tau}},h_{\mathfrak{r}_{\tau}})$.
\end{mdframed}
\end{nobreaks}

\begin{Remark}
By construction the number of samples  until the algorithm stops is smaller for the overlapping than the non-overlapping case if the same split points are used. In this case, for equal-width seMCD-buckets, the output bucket will usually be wider for the overlapping case. Alternatively, one can use the same buckets as in the non-overlapping case and add overlapping ones. Then, the seMCD-algorithm for the overlapping buckets is still guaranteed to stop before the corresponding non-overlapping one. For equal-width seMCD-buckets, the output bucket will then typically have the same widths in both cases. 
\end{Remark}

    \subsubsection{Statistical properties of the seMCD-algorithm}\label{sec_stat_guar}
    
    We are now ready to explain why the statistical guarantees carry over to the situation with more than two buckets both for the non-overlapping and overlapping case:

\begin{Lemma}\label{lemma_new_FWER}
Let the monotonicity condition \ref{ass_Monotonnicity_boundary_sequences} hold.
Consider a set of split points $h_1, \ldots, h_k$ for the overlapping or non-overlapping seMCD-algorithm, and denote the true (unknown) value by $h$.  Denote by $\sB=\sB(\bm{\xi}_1,\bm{\xi}_2,\ldots)$ the seMCD-bucket that the algorithm outputs. Then,	      
\begin{align*}
&\P_h(h\in \sB)\ge \P_h\left( L_N^{(h)}<S_N<U_N^{(h)} \text { for all }N \in\mathbb{N}\right).
\end{align*}
\end{Lemma}

\begin{proof}[Proof of Lemma \ref{lemma_new_FWER}]
	By the monotonicity condition \ref{ass_Monotonnicity_boundary_sequences} it holds, with the notation $[k]=\{1,\ldots,k\}$,
\begin{align*}
    &\P_h(h \in \sB) = 1-\P_h(h \notin \sB)\\ 
    &\geq 1-\P_h\left(\left( \bigcup_{\{\ell \in [k]: h_\ell\geq h\}} \left\{\exists N \in \mathbb N: S_{N}\geq U_{N}^{(h_\ell)} \right\} \right) \cup \left( \bigcup_{\{\ell \in [k]: h_\ell\leq h\}} \left\{\exists N \in \mathbb N: S_{N}\leq L_{N}^{(h_\ell)}  \right\} \right) \right)\\
    &\geq 1-\P_h\left( \left\{\exists N \in \mathbb N:S_{N}\geq U_{N}^{(h)}\right\}  \cup  \left\{\exists N \in \mathbb N: S_{N}\leq L_{N}^{(h)} \right\} \right)\\
    &= \P_h\left( L_{N}^{(h)} < S_{N}< U_{N}^{(h)} \text{ for all } N \in \mathbb N \right). 
\end{align*}
\end{proof}

Lemma \ref{lemma_new_FWER} shows, in particular, that probabilistic statements as in \eqref{I}  with respect to the split points $h_j$, $j=1,\ldots,k$, carry over to the desired statistical guarantees of the seMCD-algorithm as long as the monotonicity condition \ref{ass_Monotonnicity_boundary_sequences} is satisfied. This is the case for the non-parametric boundary sequences discussed in Section~\ref{sec_sequential_construction}.

The following result gives a lower bound for the probability that the algorithm stops in finite time, which we will show to be equal to one for the boundary sequences in Section~\ref{sec_sequential_construction}.

\begin{Lemma}\label{lem_new_stop}
Denote by $\tau=\tau(\bm{\xi}_1,\bm{\xi}_2,\ldots)$ the stopping time of the seMCD-algorithm, i.e.\ the number of steps until the algorithm terminates and outputs a seMCD-bucket.
Let $\wideunderbar{h}\le h\le \widebar{h}$ and $(\wideunderbar{h},\widebar{h})$ be one of the seMCD-buckets,  then
\begin{align*}
	\P_h(\tau<\infty)&\ge \P_h\left(\liminf_{N\to\infty}(S_N-U_N^{(\wideunderbar{h})})>0\right)+\P_h\left(\liminf_{N\to\infty}(L_N^{(\widebar{h})}-S_N)>0\right)-1.
\end{align*}
\end{Lemma}
In particular, if the two probabilities on the right-hand side are equal to one, then the stopping time is almost surely finite.

\begin{proof}
By construction, the algorithm will be guaranteed to stop if the seMCD-process is above $U_N^{(\wideunderbar{h})}$ and at the same time below $L_N^{(\widebar{h})}$. If $\liminf_{N\to\infty}(S_N-U_N^{(\wideunderbar{h})})>0$, then $S_N $ will be above $U_N^{(\wideunderbar{h})}$ for all $N$ large enough. Similarly, if $\liminf_{N\to\infty}(L_N^{(\widebar{h})}-S_N)>0$, the seMCD process $S_N$ will  be below $L_N^{(\widebar{h})}$ for all $N$ large enough. Consequently,
    \begin{align*}
	\P_h(\tau<\infty)&\ge \P_h\left(\liminf_{N\to\infty}(S_N-U_N^{(\wideunderbar{h})})>0,\;\liminf_{N\to\infty}(L_N^{(\widebar{h})}-S_N)>0\right).
\end{align*}
By $P(A\cap B)= P(A)+P(B)-P(A\cup B)\ge P(A)+P(B)-1$, the assertion follows.
\end{proof}

\begin{Remark}\label{rem_simplicial}
   Despite this stopping guarantee   
   the methodology does not overcome the curse of dimensionality associated with some depth functions, such as the simplicial depth. Indeed, the simplicial depth, computed with respect to absolutely continuous distributions, that are angularly symmetric about the origin, contracts to 0 with rate $O(1/2^d)$, when the dimension $d$ grows \citep[Theorem 4]{LiuSimplicial}. In particular, for $d\ge 5$, all simplicial depth values are included in a (hypothetical) seMCD-depth bucket covering $[0,0.1]$. So, while such a bucket is still meaningful for $d=1$, this is no longer the case for $d\ge 1$. \\ 
   To guarantee a similar precision for each dimension $d$, the split points need to be scaled by the factor $1/2^{d-1}$. However, with such rescaled buckets more seMCD-samples are required in higher dimensions. In fact, some preliminary simulations 
   with $d=5$ and overlapping buckets with non-rescaled split points of $0.05,0.1,0.15,\ldots, 0.6$ revealed that the samples required until the algorithm stops were of the order of magnitude of $10^3/2$ larger for the rescaled buckets than for the non-rescaled ones.
\end{Remark}

\subsubsection{Greedy seMCD-algorithm}\label{sec_another_look_algorithm}
In this section, we introduce an alternative algorithm, that has both advantages and disadvantages over the previously described approach.
 Here, we keep track of the  feasible output seMCD-buckets by tracking the largest feasible lower index $\mathfrak{l}$ as well as the smallest feasible upper index $\mathfrak{r}$, with the meaning that all buckets contained in $(h_{\mathfrak{l}},h_{\mathfrak{r}})$ are still feasible output buckets. The algorithm terminates as soon as only one bucket, or the non-empty intersection of two buckets, is contained in this interval and outputs this interval. 
We call the algorithm greedy as, due to this tracking, the number of checks necessary to update $\mathfrak{r}$ and $\mathfrak{l}$ after an additional sample is taken is smaller than in the previous algorithms. Indeed, the maximum/minimum are only taken over still feasible indices $\mathfrak{l}<j<\mathfrak{r}$ in the greedy algorithm, while they are taken over all $j$ in the previous algorithms.

\vspace{2mm}
\begin{mdframed}
\textbf{\underline{Greedy seMCD-algorithm}}

\renewcommand{\theenumi}{\arabic{enumi}}
\renewcommand{\labelenumi}{\textbf{Step} \textbf{\theenumi}:}
\begin{enumerate}[align=left,leftmargin=*,widest=10]
	\item Set the lower feasible index $\mathfrak{l}=0$ and the upper feasible index to $\mathfrak{r}=k+1$.
	\item Take an additional sample and check in the following order
		\begin{itemize}
			\item[1.] if, for some $\mathfrak{l}< j <\mathfrak{r}$, it holds $S_N\ge U_N^{(h_j)}$:\\[2mm] Update $\mathfrak{l} := \max\{\mathfrak{l}< j <\mathfrak{r}: S_N\ge U_N^{(h_j)}\}$. \\[2mm]
            If $\mathfrak{r}-\mathfrak{l}=1$ in the non-overlapping case, resp.\ $\mathfrak{r}-\mathfrak{l}\le 2$ in the overlapping case,\\ go to Step~3.\\[1mm]
			\item[2.] if, for some $\mathfrak{l}< j <\mathfrak{r}$, it holds $S_N\le L_N^{(h_j)}$:\\[2mm] Update $\mathfrak{r} := \min\{\mathfrak{l}< j <\mathfrak{r}: S_N\le L_N^{(h_j)}\}$.\\[2mm] If $\mathfrak{r}-\mathfrak{l}=1$ in the non-overlapping case, resp.\ $\mathfrak{r}-\mathfrak{l}\le 2$ in the overlapping case,\\ go to Step~3.
            \end{itemize}
			\item Output the seMCD-bucket $(h_{\mathfrak{l}},h_{\mathfrak{r}})$.
		\end{enumerate}
\end{mdframed}

The greedy algorithm produces a nested sequence of intervals $(h_{\mathfrak{l}},h_{\mathfrak{r}})$ which contain the true value uniformly with a probability of at least $1-\alpha$.   It stops as soon as this interval equals one of the prespecified seMCD-buckets or their non-empty intersection. 

The greedy algorithm stops earlier than the previous algorithm: Indeed, it can already stop within the gray area as in Figure~\ref{fig_nonstopping_regions}, for example, if the seMCD-process in (a) first crosses above the upper dash-dotted orange line and then -- without leaving the gray area -- below the lower dashed blue line.  In fact, if $\wideunderbar{h}< h<\widebar{h}$ for the true value $h$ and $(\wideunderbar{h},\widebar{h})$ is one of the seMCD-buckets, then the greedy algorithm stops as soon as the seMCD-algorithm has at least once been above the upper boundary sequence associated with $\wideunderbar{h}$ and below the lower boundary sequence associated with $\widebar{h}$. Consequently, the lower bound in Lemma~\ref{lem_new_stop} can be replaced by:
\begin{align*}
	\P_h(\tau_{\operatorname{greedy}}<\infty)&\ge 
	 \P_h\left(S_{\wideunderbar{N}}\ge U_{\wideunderbar{N}}^{(\wideunderbar{h})}\text{ for some }\wideunderbar{N}\ge 1\right)+\P_h\left(S_{\widebar{N}}\le L_{\widebar{N}}^{(\widebar{h})}\text{ for some }\widebar{N}\ge 1\right)-1.
\end{align*}
This comes at the cost of a higher probability of outputting the wrong bucket: For example, in Figure~\ref{fig_nonstopping_regions} the seMCD-process with true value in $(0.1,0.25)$ could first falsely cross the upper dashed-dotted blue while remaining between the green lines, then fall below the lower dashed green line to falsely output the bucket $(0.25,0.4)$ before entering the white area between the blue and orange lines, where the original algorithm would correctly output the bucket $(0.1,0.25)$. While this happens with low probability and does not affect the algorithms statistical guarantees, it is most likely to occur at the beginning, where each new sample has a larger impact on the seMCD-process  and the curves are still close together, forming a larger gray area. In other words, the original algorithm is more conservative, especially at the beginning of the sampling. This is a desirable property -- in particular in combination with the non-parametric asymptotic boundary sequences as discussed in Section~\ref{sec_sequential_construction}. However, the assertion and proof of Lemma~\ref{lemma_new_FWER} and thus the corresponding statistical guarantee also hold for the greedy algorithm.

\subsection{SeMCD-buckets for monotonic E-depths}\label{sec_transformation}
In Section~\ref{type_b} monotonic E-depths were introduced which are of the form
\begin{align*}
   g:= F(h):=F(\E(H(\bm{\xi})))
\end{align*}
for a strictly monotonic function $F$. In the following, we will consider $F$ to be strictly increasing. The less common strictly decreasing case is analogous. The seMCD-methodology based on such monotonic E-depths is applicable for the calculation of depth buckets as in Section~\ref{section_depth_level_sets} or anomaly detection as in Section~\ref{section_outlier} by choosing non-overlapping depth buckets 
\begin{align*}
(-\infty,g_1),(g_1,g_2),\ldots,(g_l,\infty), \qquad  \text{with } \quad g_1<g_2<\ldots<g_l\quad \text{ and }\quad g_j=F(h_j).
\end{align*}
For the overlapping case, one can similarly use $   (-\infty,g_2), (g_1,g_3), \ldots, (g_{l-2},g_l), (g_{l-1},\infty).$
This is because, for $F$ strictly increasing and any $\ell<u$,
\begin{align*}
   h_{\ell} < \E (H(\bm{\xi})) < h_{u}\quad \iff \quad g_{\ell} <F(\E (H(\bm{\xi})))<g_u.
\end{align*}
Consequently, one can run the usually seMCD-algorithm with the seMCD-buckets based on $h_j$, $j=1,\ldots,l$,  for $\E (H(\bm{\xi}))$ and then  output the corresponding bucket based on $g_j=F(h_j)$, $j=1,\ldots,l$, for $F(\E (H(\bm{\xi})))$ without losing the statistical guarantees.

For binary maximum depth classification in combination with E-depths we have proposed to use the seMCD-algorithm directly on the difference of the two empirical depths as in \eqref{eq_classification_difference}. For the monotonic E-depths, this difference reads with the same notation as in Section~\ref{subsec_classification}
\begin{align}\label{eq_diff_modified_E_class}
    F(\E[ H_{\bx_1,\ldots,\bx_{r_x}}(\bm \zeta)])-F(\E[H_{\by_1,\ldots,\by_{r_y}}(\bm \zeta)]).
\end{align}
In combination with the non-overlapping seMCD-buckets $(-\infty,0),(0,\infty)$, the methodology is applicable by using the same seMCD-buckets on $\E[ H_{\bx_1,\ldots,\bx_{r_x}}(\bm \zeta)-H_{\by_1,\ldots,\by_{r_y}}(\bm \zeta)]$ because this expected value is negative (positive) iff the value in \eqref{eq_diff_modified_E_class} is negative (positive). Thus, the statistical guarantee remains. However, because of the non-linearity of $F$, using an additional overlapping block $(-\epsilon_F,\epsilon_F)$ for the quantity in \eqref{eq_diff_modified_E_class} requires of an additional assumption. Let us assume that $F$ is  Lipschitz-continuous with Lipschitz constant $L_F>0$. Then, using the seMCD-algorithm with the third bucket of $(-\epsilon_F/L_F,\epsilon_F/L_F)$ on $\E[ H_{\bx_1,\ldots,\bx_{r_x}}(\bm \zeta)-H_{\by_1,\ldots,\by_{r_y}}(\bm \zeta)]$ gives the desired statistical guarantee for the quantity \eqref{eq_diff_modified_E_class} with the buckets $(-\infty,0),(-\epsilon_F,\epsilon_F),(0,\infty)$. This is because
\begin{align*}
  \left|\E[ H_{\bx_1,\ldots,\bx_{r_x}}(\bm \zeta)-H_{\by_1,\ldots,\by_{r_y}}(\bm \zeta)]\right|\le \frac{\epsilon_F}{L_F}\
     \implies  
  \left |F( \E[ H_{\bx_1,\ldots,\bx_{r_x}}(\bm \zeta)])-F(\E[H_{\by_1,\ldots,\by_{r_y}}(\bm \zeta)])\right|\le \epsilon_F.
\end{align*}
Note that, depending on the true underlying depth values (not only the difference), this might be a very conservative procedure in terms of the true coverage for the $(-\epsilon_F,\epsilon_F)$-bucket.

\section{Construction of the boundary sequences}\label{sec_boundary}
By Lemma~\ref{lemma_new_FWER} the algorithms described in the previous section require lower and upper boundary sequences that fulfil \eqref{I}, in order to keep their statistical guarantees. In constructing such boundary sequences, we will make use of the following connection to sequential testing: In the latter, the sample size is not fixed in advance. Instead, data collection continues until a stopping criterion is reached. As with classical testing, the probability of a type-I-error is controlled at a fixed level $\alpha$, that is, if the null hypothesis is true, the procedure will stop with a probability of at most $\alpha$. Typically, for a sequence of iid\ random variables $Y_1,Y_2,\ldots$, with $h=\E (Y_1)$ and a null hypothesis of $h=h_1$, the test decision is based on the partial sum process $S_N=\sum_{i=1}^NY_i$ in addition to the appropriate boundary sequences $L_N^{(h_1)}$ and $U_N^{(h_1)}$. The sampling stops as soon as $S_N \not\in (L_N^{(h_1)},U_N^{(h_1)})$. Clearly, the outlined sequential test with $Y_i=H(\bm{\xi}_i)$ controls the type-I-error at level $\alpha$ iff \eqref{I} holds.
    
If the distribution of the seMCD process $\{S_N\}$ depends only on the unknown quantity of interest $h=\E (H(\bm{\xi}))$, but is otherwise known, then the probabilities in \eqref{I} are in principle known, permitting the construction of such boundary sequences. However, other scenarios are possible: It might happen that the distribution of the seMCD-process is known but not analytically accessible, that it depends additionally on other unknown quantities, or that it is indeed not known at all. To elaborate, in the context of depth calculation as discussed in this paper, we encounter the following cases:
\renewcommand{\theenumi}{\arabic{enumi}}
\renewcommand{\labelenumi}{\textbf{Case} \textbf{\theenumi}:}
\begin{enumerate}[wide]
	\item \label{Case1}The distribution of $H(\bm{\xi})$ is fully known up to the unknown expectation $h$, with all the probabilities in \eqref{I} being easily computable. In the context of sequential testing, this corresponds to a two-sided simple hypothesis in a fully parametric situation.  As an example, consider Type A depths with indicator kernels such as the simplicial depth, see Appendix~\ref{subsec_multivariate_depths_A}, where $H(\bm{\xi})\sim \textrm{Bernoulli}(h)$ with unknown success probability $h$.  In Section~\ref{sec_gandy_approach} we describe how to obtain boundary sequences in this Bernoulli case that fulfil \eqref{I} -- \eqref{III}.
	\item \label{Case2} Similarly to Case~\ref{Case1}, the  distribution of $H(\bm{\xi})$ is theoretically known, but an exact analytical evaluation is difficult or impossible. Examples are given by the integrated depths, see \eqref{DI}, where the expectation is with respect to a distribution that is not related to the one with respect to which the depth is computed. Without access to this distribution it is practically impossible to construct exact boundary sequences fulfilling \eqref{I}, see Remark~\ref{rem_generalization}. However, the asymptotic approach as outlined in Section~\ref{sec_sequential_construction} is applicable.
	\item \label{Case3} The distribution of $H(\bm{\xi})$ is known, but depends not only on the unknown parameter $h$ but also on other unknown quantities. As an example, consider binary maximum depth classification in combination with a Type A depth with indicator kernel as in \eqref{eq_classification_difference}. The distribution of the corresponding $H(\bm{\xi})$ is a three-point distribution on $\{-1,0,1\}$ depending on the two unknown empirical depths $d_{\bx_1,\ldots,\bx_{r_{\bx}}}$ and $ d_{\by_1,\ldots,\by_{r_{\by}}}$ with respect to each of the two samples and not only on $h=d_{\bx_1,\ldots,\bx_{r_{\bx}}}-d_{\by_1,\ldots,\by_{r_{\by}}}$. In sequential testing, this corresponds to a test situation with a composite rather than a simple hypothesis $d_{\bx_1,\ldots,\bx_{r_{\bx}}}=d_{\by_1,\ldots,\by_{r_{\by}}}$. As for Case~\ref{Case2} this makes constructing exact boundary sequences challenging, see Remark~\ref{rem_generalization}. However, the asymptotic approach of Section~\ref{sec_sequential_construction} is applicable.
	\item\label{Case4} The distribution of $H(\bm{\xi})$ is not known, and there is no reasonable choice of a parametric statistical model for it, e.g.\ because the distribution depends on the entire unknown distribution of the observations and not only on the unknown parameter $h$. An example is  the modified band depth, see Appendix~\ref{subsec_functional_depth_A}. This corresponds to a non-parametric testing situation with no knowledge of the underlying distributional family. Without access to the probabilities in \eqref{I}, it is impossible to obtain exact boundary sequences based on $S_N$.  However, as is often the case in non-parametric statistics, the asymptotic approach described in Section \ref{sec_sequential_construction} is applicable. 
\end{enumerate}
 
\subsection{Boundary sequences for Type A depth functions with  indicator kernel}\label{sec_gandy_approach}
In this section, we first discuss the construction of boundary sequences if $H(\bm{\xi})\sim \mbox{Bernoulli}(h)$, with $0<h<1$, followed by the statistical properties of the corresponding seMCD-algorithm. Several sequential tests have been developed in this context dating from the 1970s \citep{Robbins1970,Lai1976} to today \cite{Fischer2024B}. In this article we follow the approach suggested by \cite{Gandy2009}. The latter paper was the basis for the introduction of $p$-value buckets discussed in \cite{Gandy2020b}, which inspired this paper.

\subsubsection{Recursive computation of the boundary sequences}\label{rec_comp_seq}
The boundary sequences in this section are based on an error spending sequence $(\alpha_N)_{N \in \mathbb N}$, where the construction guarantees that $\alpha_N$ is an upper bound for the probability of making a wrong decision within the first $N$ sampling steps. Thus, $(\alpha_N)_{N \in \mathbb N}$ has to be nondecreasing with values in $[0,\alpha]$ and $\alpha_N \to \alpha$ as $N \to \infty$, see e.g.\ \cite{LanDeMets1983} for details.  The spending sequence determines the shape of the boundary sequences. As an example, consider the class of spending sequences $\alpha_N = \alpha \cdot N/(\kappa+N)$, where the parameter $\kappa \in \mathbb N$ denotes the number of iterations, after which half of the error is spent, i.e.\ $\alpha_{\kappa} = \alpha/2$.  

For the Bernoulli case, clearly $S_N$ only takes integer values, and it holds $0\le S_{N+1}-S_N\le 1$ for all $N$. Similarly, our boundary sequences will be constructed to take only integer values and such that $0\le L_{N+1}-L_N\le 1$ as well as $0\le U_{N+1}-U_N\le 1$ for all $N$, where we suppress the upper index $(h)$ of the boundary sequences for notational ease in the description of the algorithm. Thus, at the time of stopping, the seMCD-process will be equal to one of the boundary sequences. Consequently, for the Bernoulli case the inequalities in the conditions for the calculation of $\mathfrak{l}_N, \mathfrak{r}_N$ in the seMCD-algorithm can be replaced by equalities.

We will now explain how to recursively compute boundary sequences for a given error spending sequence and split point $h$:
With an initialisation of $S_0=0, L_0=-1, U_0=1$, define 
\begin{align*}
	&\mathfrak p_{N}(k) := \P_h(S_{N}=k, L_j < S_j < U_j \text{ for all }j< N),
\end{align*}
i.e.\ the probability of $S_N=k$, $k\in \mathbb{N}$, and not having previously hit a boundary sequence.  For $N=0$, clearly,  $\mathfrak p_0(0)=1$ and the following recursions hold:
\begin{align*}
&\mathfrak p_{N+1}(U_N)=\mathfrak p_{N}(U_N-1)\cdot h,\qquad
\mathfrak p_{N+1}(L_N+1)=\mathfrak p_{N}(L_N+1)\cdot (1- h),
\\
&\mathfrak p_{N+1}(k)=\mathfrak p_{N}(k)\cdot (1-h)+\mathfrak p_{N}(k-1)\cdot h \quad \text{for }L_N+1< k < U_N.
\end{align*}
In addition, define the probability of making a wrong decision within the first $N$ sampling steps by first hitting the \textbf{upper} boundary sequence, i.e.\
 \[ 
\beta_U{(N)} := \P_h(\text{For some } \ell \le N: S_\ell = U_\ell \text{ and }  L_j<S_j<U_j \text{ for all }j< \ell),
\]
as well as the corresponding probability for first hitting the \textbf{lower} boundary sequence by 
\[ 
\beta_L{(N)} := \P_h(\text{For some } \ell \le N: S_\ell = L_\ell \text{ and } L_j< S_j< U_j \text{ for all }j< \ell).
\]
Clearly, the following recursions hold:
\begin{align}\label{eq_rec_betas}
    \beta_U{(N+1)} = \beta_U{(N)} + \mathfrak p_{N+1}(U_{N+1}), \qquad \beta_L{(N+1)} = \beta_L{(N)} + \mathfrak p_{N+1}(L_{N+1}).
\end{align}

The below algorithm constructs the boundary sequences in such a way that they fulfil 
\begin{align}\label{eq_guar_betas}
\beta_U{(\ell)} \leq \frac{\alpha_{\ell}}{2} \quad \text{and}\quad \beta_L{(\ell)} \leq \frac{\alpha_{\ell}}{2} 
\end{align}
for all $\ell\in\mathbb{N}$. The three quantities $\mathfrak p_{N+1}(k)$, $\beta_U(N)$ as well as $\beta_L(N)$ are only based on the boundary sequences for $j\le N$. If $\beta_U(N)$ and $\beta_L(N)$ fulfil \eqref{eq_guar_betas} for $\ell=N$ (induction hypothesis), then, based only on these three quantities, the algorithm chooses the boundary sequences at the next point $N+1$ in such a way, that   \eqref{eq_guar_betas} also holds for $\ell=N+1$ (induction step). Consequently, \eqref{eq_guar_betas} holds for all $\ell\in\mathbb{N}$ by induction, as clearly, for $N=0$,  $\beta_U{(0)}=0\le \alpha_0/2$ and $\beta_L{(0)}=0\le \alpha_0/2$.

Let us explain the construction of $U_{N+1}$ in detail: Because the error spending sequence is nondecreasing and as \eqref{eq_guar_betas} is guaranteed for $\ell=N$, by induction hypothesis, it holds $\alpha_{N+1}/2 \ge \beta_U{(N)}$. If $\mathfrak p_{N+1}(U_N) \le \alpha_{N+1}/2 - \beta_U{(N)}$,  then for $U_{N+1}:=U_N$  it indeed holds by \eqref{eq_rec_betas}
\begin{align*}
    \beta_U{(N+1)} = \beta_U{(N)} + \mathfrak p_{N+1}(U_{N})\le \alpha_{N+1}/2.
\end{align*}
Otherwise, choosing $U_{N+1}:=U_N+1$ yields $\beta_U{(N+1)}=\beta_U{(N)}$, such that the first assertion in \eqref{eq_guar_betas} holds by monotonicity  of the error spending sequence. 
The construction of $L_{N+1}$ is similar, see the below algorithm for details,  and analogously yield the second assertion in \eqref{eq_guar_betas}. 

We summarize the above findings in the following algorithm:

\vspace{2mm}
\begin{mdframed}
\textbf{\underline{Algorithm to compute the boundary sequences}}
\vspace{2mm}

\textbf{Initialise:} $L_0=-1$, $U_0=1$, $p_0(0)=1$, $\beta_U(0)=0=\beta_L(0)$.\\[2mm]
\textbf{Recursively, compute} the values for $N+1$ from those of $N$ as follows:
\begin{itemize}
\item Compute $\mathfrak p_{N+1}(U_N)=\mathfrak p_{N}(U_N-1)\cdot h$.
\begin{align*}
            &\textbf{If }\mathfrak p_{N+1}(U_N)\leq \frac{\alpha_{N+1}}2 - \beta_U(N): \qquad \text{Set } U_{N+1} = U_N.\\[1mm]             &\textbf{Otherwise:}\qquad \text{Set }U_{N+1} = U_N+1.
        \end{align*}
\item Compute $\mathfrak p_{N+1}(L_N+1)=\mathfrak p_{N}(L_N+1)\cdot (1-h)$.
\begin{align*}
            &\textbf{If }\mathfrak p_{N+1}(L_N+1)\leq \frac{\alpha_{N+1}}2 - \beta_L(N): \qquad \text{Set } L_{N+1} = L_N+1.\\[1mm]             &\textbf{Otherwise:}\qquad \text{Set }L_{N+1} = L_N.
        \end{align*}
\item Compute $\mathfrak p_{N+1}(k) = \mathfrak p_{N}(k) \cdot (1-h) + \mathfrak p_{N}(k-1) \cdot h$ for $L_N+1< k< U_N$.
\end{itemize}
\end{mdframed}

\begin{Remark}\label{rem_generalization}
When $H(\bm{\xi})$ does not follow a Bernoulli-distribution several obstacles arise for a generalisation of the above algorithm, even if the distribution class is known as in Case~\ref{Case1}: If $H(\bm{\xi})$ is no longer a.s.\ non-negative, $S_N$ no longer increases a.s., greatly complicating things. As soon as the underlying probability distribution is continuous, keeping track of meaningful alternatives of $\mathfrak p_N(\cdot)$ becomes difficult. If the probabilities $\mathfrak p_N(\cdot)$ cannot be calculated analytically as in Case \ref{Case2}, the difficulties are obvious. When there exist multiple parameters leading to the same $h=\E (H(\bm{\xi}))$, as in Case \ref{Case3}, $\mathfrak p_N(\cdot)$ will typically depend on those underlying parameters and needs to be tracked for all of these usually infinitely many parameters,  which can be challenging. Finally, in Case \ref{Case4} there is no possibility of calculating $\mathfrak p_N(\cdot)$ (nor a suitable alternative in the continuous case) rendering such an approach impossible.
\end{Remark}

\subsubsection{Statistical Guarantees}
By the construction detailed in the previous Section \ref{rec_comp_seq}, inductively, \eqref{eq_guar_betas} holds for all $\ell\in\mathbb{N}$, such that for the true value $h$, the probability of $\{S_j\}$ having hit one of the boundary sequences up to sample $N$ is given by $\beta_U{(N)}+ \beta_L{(N)} \leq \alpha_{N}$ for all $N \in \mathbb N$. Furthermore, reintroducing the dependence on  $h$ again, by monotonicity of the probabilities with respect to $h$, it holds for any $h_-\le h_+$ and all $N$, that 
\begin{align*}
    \beta_U^{(h_-)}(N)\le \beta_U^{(h_+)}(N)\quad \text{and}\quad \beta_L^{(h_-)}(N)\ge \beta_L^{(h_+)}(N).
\end{align*}
In summary, we have proven the following lemma using again the upper indices in the boundary sequences for clarity:

\begin{Lemma}\label{lem_bernoulli-guarantee}
In the above situation with boundary sequences  constructed with the above algorithm, \eqref{I} holds as well as
\[
    \P_h \left(U^{(h)}_{\ell} < S_{\ell}< L^{(h)}_{\ell}\text{ for all }\ell\in \{1,\ldots, N\}\right) \geq 1-\alpha_{N}\quad \text{for all }N \in \mathbb N. 
\] 
Defining the stopping time $\tau^{(h)} = \inf\{N \in \mathbb N: S_N \geq U_N^{(h)} \text{ or } S_N \leq L_N^{(h)} \}$ it holds for any $h_-\le h\le h_+$
\begin{align}\label{eq_stopping_probability}
	&\P_{h_-}\left(\tau^{(h)} < \infty,S_{\tau^{(h)}}\geq U_{\tau^{(h)}}^{(h)}  \right)  \leq \frac{\alpha}{2} \quad\text{and}\quad \P_{h_+}\left(\tau^{(h)} < \infty,S_{\tau^{(h)}}\leq L_{\tau ^{(h)} }^{(h)}\right)\leq  \frac{\alpha}{2}.
\end{align}
\end{Lemma}

If \eqref{I} as well as the monotonicity condition \ref{ass_Monotonnicity_boundary_sequences} hold, then Lemma~\ref{lemma_new_FWER} delivers the desired statistical guarantee of making a wrong decision with a probability less than the chosen tolerance parameter $\alpha$. While Lemma~\ref{lem_bernoulli-guarantee} shows that \eqref{I} holds, the monotonicity condition \eqref{eq_Monotonnicity_boundary_sequences} is not necessarily fulfilled. Therefore, here, for completeness, we will repeat arguments of \citet[Theorem 3]{Gandy2020b},  making use of \eqref{eq_stopping_probability}. 

\begin{Corollary}
In the Bernoulli-case with boundary sequences constructed by the above algorithm, which fulfil
\begin{align}\label{MonStrong}        U_N^{(h_1)} \leq U_N^{(h_2)} \leq ... \leq U_N^{(h_{l})} \quad\text{as well as}\quad L_N^{(h_1)} \leq L_N^{(h_2)} \leq ... \leq L_N^{(h_{l})} \tag{\textbf{Mon}$\bm{(h_1, \ldots,h_k)}$}
\end{align}
for all split points $h_1, \ldots,h_k$, with the notation of Lemma~\ref{lemma_new_FWER}, it holds
\begin{align*}
    \P_h(h\in \sB)\ge 1-\alpha.
\end{align*}
\end{Corollary}

The monotonicity condition \eqref{MonStrong} is weaker than \ref{ass_Monotonnicity_boundary_sequences} and
can be easily checked when computing the boundary sequences using the recursive algorithm above. 
Indeed, \citet[Lemma 2]{Gandy2020b} entails that, if $\alpha \leq 1/4$, there exists a constant $n_0 \in \mathbb N$ such that \eqref{MonStrong} holds for all $N \geq n_0$. If the condition does not hold for lower values, the sequences can be adjusted so that the monotonicity is also guaranteed for $N < n_0$.

\begin{proof}
First, consider the case, where both $\mathfrak U = \min \{h_j: h_j \geq h\}$ and $\mathfrak L = \max \{h_j: h_j \leq h\}$ exist, then by \eqref{MonStrong} and \eqref{eq_stopping_probability} for all $0< h< 1$
\begin{align}\label{eq_FWER_1}
    & \P_{h}\left( \bigcup_{h_j\geq h} \left\{\tau^{(h_j)} < \infty,S_{\tau^{(h_j)}}\geq U_{\tau^{(h_j)}}^{(h_j)}\right\}  \right)= \P_{h}\left( \tau^{( \mathfrak U)} < \infty,S_{\tau^{(\mathfrak U)}}\geq U_{\tau^{(\mathfrak U)}}^{( \mathfrak U)}\right) \leq \frac{\alpha}{2},\\
    & \P_{h}\left( \bigcup_{h_j\leq h} \left\{\tau^{(h_j)} < \infty,S_{\tau^{(h_j)}}\leq L_{\tau^{(h_j)}}^{(h_j)}\right\}  \right)= \P_{h}\left( \tau^{( \mathfrak L)} < \infty,S_{\tau^{(\mathfrak L)}}\leq L_{\tau^{(\mathfrak L)}}^{( \mathfrak L)}\right) \leq \frac{\alpha}{2}.\label{eq_FWER_2}
    \end{align}
    This yields the assertion. In the case, where $\mathfrak U$ or $\mathfrak L$ do not exist, the unions on the left-hand side of \eqref{eq_FWER_1} and \eqref{eq_FWER_2} are empty, such that the corresponding probabilities are zero, and the assertion follows.
 \end{proof}

Finally, the algorithm will stop with probability one in finite time, if
the true value is inside one of the buckets.  In fact, \citet[Theorem 1]{Gandy2009} not only shows that the stopping time is finite in this case, but also  that it has finite expectation.
   
\subsection{General boundary sequences without parametric assumptions}\label{sec_sequential_construction}\label{subsec_sequential_asymptotic_boundaries}
The algorithm in the previous section only works for $H(\bm{\xi})$ Bernoulli distributed, which is the situation in very specific instances of Case \ref{Case1} such as simplicial depth, spherical depth or lens depth. As pointed out in Remark~\ref{rem_generalization}, obtaining similar algorithms may be difficult, or even impossible as in the non-parametric setting of Case \ref{Case4}. To solve this, we discuss  a different non-parametric 
construction method for the boundary sequences, where the statistical guarantees are  obtained in a suitable asymptotic sense. For that we make use of a non-parametric sequential testing approach recently introduced by~\cite{gnettner2025newflexibleclasssharp}. This approach has its origins in the sequential change point analysis of time series and was initially introduced by \cite{Chu1996} and  later extended by \cite{Horvath2004}. A detailed review on the progress in this field is provided by \cite{ClaudiaAlex2023}.

Recall, that the summands $H(\bm{\xi}_j)$ of the seMCD-process in \eqref{eq_SM} are iid with unknown mean $h=\E(\bm{\xi}_j)$, which is the quantity of interest. In this subsection, we will not make any assumptions on the distribution of $H(\bm{\xi}_1)$, but we do require second moments and non-degeneracy, i.e.\  $0<\sigma^2:=\V(H(\bm{\xi}_1))<\infty$ with typically unknown variance $\V(H(\bm{\xi}_1))$. The latter condition is typically fulfilled in a depth context. Thus, the methodology derived in this section is applicable in all situations discussed in Section~\ref{section_2_depth} as long as second moments exist. This comes at the cost of replacing the finite-sample guarantee with respect to the tolerance in Section~\ref{sec_gandy_approach} by asymptotic guarantees.

To explain the motivation behind the asymptotics,  we first replace the finite-sample guarantee in \eqref{I} by
\begin{align*}
    \lim_{m\to\infty}\P_{h}\left( L_{N,m}^{(h)}(\sigma^2)<S_N<U_{N,m}^{(h)}(\sigma^2) \text { for all }N\ge m\right)= 1-\alpha.
\end{align*}
The parameter $m$ plays a similar  role here as the sample size in classical statistics and can be considered as a burn-in period for the asymptotic approximations to work well. The boundary sequences depend on $m$ in a similar way as the normalization of the mean in a $t$-test depends on the sample size. The burn-in parameter $m$ has to be chosen in advance but will typically be much smaller than standard choices for Monte Carlo repetitions. In our simulation study, we choose $m=500$.

By a careful mathematical analysis it is possible to construct boundary sequences for $N<m$ without violating the above asymptotic guarantee. Thus,  it is possible to stop before $m$ samples have been taken, as long as the evidence is strong enough. Furthermore, in practice, we need to replace the usually unknown variance $\sigma^2$ with a suitable estimator $\tilde{\sigma}_N^2$ based only on the samples $H(\bm{\xi}_1),\ldots, H(\bm{\xi}_N)$ taken so far. In summary, we  construct boundary sequences, such that  the following asymptotic statement holds 
\begin{align}
    \lim_{m\to\infty}\P_{h}\left( L_{N,m}^{(h)}(\tilde{\sigma}_N^2)<S_N<U_{N,m}^{(h)} (\tilde{\sigma}_N^2)\text { for all }N\ge 1\right)= 1-\alpha. \tag{I$^*$}\label{Iasym}
\end{align}
This statement will replace \eqref{I} in the non-parametric situation.

Indeed, \citet[Theorem~3.1]{gnettner2025newflexibleclasssharp} yields after some elementary algebra:
\begin{Theorem}\label{theorem_boundary_nonparam}
Let $\{\bm{\xi}_j\}$ be iid with $h=\E (H(\bm{\xi}_1))$ and $0<\V(H(\bm{\xi}_1))<\infty$. Then,  \eqref{Iasym} holds  for the following class of boundary sequences: 
\begin{align}\label{eq_bound_np}
   & L_{N,m}^{(h)}(\tilde{\sigma}_N^2)=N\cdot h-\tilde{\sigma}_N\cdot w_{(\gamma_1,\gamma_2)}(N,m),\quad U_{N,m}^{(h)}(\tilde{\sigma}_N^2)=N\cdot h+\tilde{\sigma}_N\cdot w_{(\gamma_1,\gamma_2)}(N,m),\\
   & \text{where }w_{(\gamma_1,\gamma_2)}(N,m)=c_{\alpha}(\gamma_1,\gamma_2)\cdot m^{\gamma_2-\frac 1 2} \cdot N^{\gamma_1}\cdot (m+N)^{1-\gamma_1-\gamma_2}\notag
\end{align}
for $0\le \gamma_1,\gamma_2<1/2$, 
 \begin{align*}
     \tilde{\sigma}_N^2=\hat{\sigma}_N^2+\mathds 1_{\{\hat{\sigma}_N^2=0\}}\cdot \infty,\quad \hat{\sigma}^2_N=\frac{1}{N-1}\sum_{j=1}^N\left(H(\bm{\xi}_j)-\frac{1}{N}\sum_{i=1}^NH(\bm{\xi}_i)\right)^2,
 \end{align*}
 with the convention that $0\cdot\infty=0$,
 and $c_{\alpha}(\gamma_1,\gamma_2)$ is the $(1-\alpha)$-quantile of $\sup_{0\le t\le 1}\frac{|B(t)|}{t^{\gamma_1}\, (1-t)^{\gamma_2}}$ for a standard Brownian bridge $\{B(\cdot)\}$.
\end{Theorem}
\citet[Theorem~3.1]{gnettner2025newflexibleclasssharp} allow for more general shape functions $w(N,m)$. We have selected the above shape functions because the critical values $c_{\alpha}$ are obtained as the quantiles of a supremum over  a weighted Brownian bridge, i.e.\ the supremum is only taken over the unit interval. Other shape functions typically require the calculation of quantiles of a supremum over the full positive line of a weighted Gaussian process, which is numerically problematic. For $\gamma_1=\gamma_2=0$ the quantiles of the limit are theoretically known, otherwise they can be approximated e.g.\ by the methodology introduced in \cite{franke2022adaptive}.

\begin{Remark}\label{rem_offset}
    In order to avoid small-sample instabilities due to the estimation of the variance, it can be beneficial to take more than a minimal number $\mathfrak l_m$ of samples or equivalently set $\tilde{\sigma}_j=\infty$ for $j\le \mathfrak l_m$, where $\mathfrak l_m$ is chosen much smaller than $m$. Mathematically, if $\mathfrak l_m/m\to 0$, the assertion of Theorem~\ref{theorem_boundary_nonparam} remains true.
\end{Remark}

The choices $\gamma_1=\gamma_2=1/2$ are excluded from Theorem \ref{theorem_boundary_nonparam}. These choices lead to $w(N,m)$ being proportional to $\sqrt{N}$ as in standard a-posteriori testing.  However, by the law of iterated logarithm, \eqref{Iasym} does not hold for these choices. The parameter $\gamma_1$ regulates the boundary sequence at the beginning, while $\gamma_2$ regulates it at infinity. Generally, a smaller $\gamma_j$, $j=1,2$, leads to more conservative boundaries at the respective end. As a compromise between being close to the a-posteriori rate at infinity but not too liberal at the beginning, we choose $\gamma_1=0.1$ and $\gamma_2=0.4$ in the simulations. In fact, for the split point $h=0.25$ this choice of $\gamma_1,\gamma_2$ with $\alpha=0.05$, $m=500$ and variance $h(1-h)$ (the latter corresponds to the Bernoulli case) is very close to the exact boundary sequences from Section~\ref{sec_gandy_approach} with $\kappa=1000$. These asymptotic boundaries are slightly more conservative for small $N$ compared to the exact ones and slightly less conservative for large $N$.

The boundary sequences in \eqref{eq_bound_np} clearly fulfil Condition~\ref{ass_Monotonnicity_boundary_sequences}. Thus, combining Lemma~\ref{lemma_new_FWER} with Theorem~\ref{theorem_boundary_nonparam} yields:

\begin{Corollary}\label{cor_fbr}
Under the assumptions of Theorem~\ref{theorem_boundary_nonparam}, it holds 
\begin{align*}
    \liminf_{m\to\infty}\P_h(h\in \sB_m)\ge 1-\alpha,
\end{align*}
where $\sB_m$ is the output seMCD-bucket of the algorithm based on the boundary sequences \eqref{eq_bound_np}.
\end{Corollary}

Similarly, the proof of \citet[Theorem~3.2]{gnettner2025newflexibleclasssharp} shows that, with the notation of Lemma~\ref{lem_new_stop}, the seMCD-process will eventually be above the upper boundary sequence associated with  $\wideunderbar{h}$ as well as below the lower boundary sequence associated with $\widebar{h}$. Consequently, in combination with Lemma~\ref{lem_new_stop}, we obtain the following stopping guarantee of the algorithm for all values $h$ that are contained in interior of one of the seMCD-buckets.

\begin{Corollary}\label{cor_47}
Under the assumptions of Theorem~\ref{theorem_boundary_nonparam}, it holds for the stopping time $\tau_m(\gamma_1,\gamma_2)$ of the seMCD-algorithm with the boundary sequences from \eqref{eq_bound_np} for $\gamma_2>0$
\begin{align*}
P_h(\tau_m(\gamma_1,\gamma_2)<\infty)=1\text{ for each $m$ and }\gamma_2>0, \qquad \lim_{m\to\infty}P_h(\tau_m(\gamma_1,0)<\infty)=1,
\end{align*}
where the true value $h$ is contained in the interior of one of the seMCD-buckets.
\end{Corollary}

In view of this result, a choice of $\gamma_2>0$ is clearly preferable.

\section{Simulation and real data analyses}\label{sec_buckets_simulations}
In this section, we present the results of an  empirical study for illustrative purposes. The calculations in this section have all been conducted with MATLAB  (R2019b, or newer) with our own implementation of depth functions. The greedy seMCD-algorithm is as in Section~\ref{sec_another_look_algorithm}, with the boundary sequences as in Theorem~\ref{theorem_boundary_nonparam} with $\gamma_1=0.1$, $\gamma_2=0.4$ and $m=500$, $\mathfrak l_m=10$  is as in Remark~\ref{rem_offset}. All empirical results are based on $1000$ repetitions.

In Section \ref{fs} we focus on a comparison of the stopping and run time for various combinations of the tuning parameters for the seMCD-algorithm.  In Section \ref{ss}, we discuss maximum depth classification with seMCD-buckets, before Section~\ref{subsec_binary_classification} illustrates the seMCD-methodology in binary classification. The source code for this section is available in the ancillary files of this arXiv preprint.

\subsection{Empirical performance of the seMCD-algorithm for overlapping buckets}\label{fs}

In this section, we compare the performance of the seMCD-algorithm with different tolerance parameters:
\[\textbf{T1 } \alpha=1\ \%,\qquad \textbf{T2 } \alpha=2.5\ \%\]
and overlapping bucket configurations: 
\begin{align*}
    &\textbf{W1 } [0,0.05), (0.025,0.075),(0.05,0.1),(0.075, 0.125),\ldots, (0.6,0.65],\\
    &\textbf{W2 } [0,0.1),(0.05,0.15),(0.1,0.2),(0.15,0.25),\ldots,  (0.55,0.65].
\end{align*}
The buckets in \textbf{W1} have a  width of $0.05$, while the ones in \textbf{W2} have a width of $0.1$. We use the modified band depth (see Appendix~\ref{subsec_functional_depth_A}) as an instance of  functional depth and the IRW depth (see Appendix~\ref{subsec_IRW_depth}) as an instance of  multivariate depth. The selection is based on the fact that the exact depth value can be computed for the first and that the literature contains  a recommended number of plain-vanilla Monte Carlo samples for the second, as we comment below.

\subsubsection{Modified band depth based on 2 curves-bands}\label{subsec_sim_mod_band_depth}
The functional modified band depth 
allows for the calculation of the exact depth value. Thus, not only will we evaluate computational properties of the seMCD-algorithm, but also calculate the empirical false-bucket-rate (\textbf{FBR}), i.e.\ the percentage of cases in which the output seMCD-bucket does not contain the true value.

To elaborate, the aim is to compute the empirical depth with respect to 
an initial fixed sample of 100 iid standard Brownian motions on the time domain $[0,1]$, which have been simulated on an equidistant grid of length 0.001. We investigate the properties of the seMCD-algorithm where the target values are the empirical depths (with respect to the above fixed sample) in three different cases, with $t\in[0,1]$. (a) $x(t)=2$,  with true empirical depth value $0.0163$ for the initial sample, (b) $x(t)=\sin(2\pi t)$,  with  true empirical depth value $0.2864$ and (c) $x(t)=B_t$, a fixed Brownian motion simulated as above, with  true empirical depth value  $0.452$.
We then run the seMCD-algorithms $1000$ times and compare the number of samples taken until stopping as well as the false-bucket-rates (\textbf{FBR}).

Table~\ref{table_new_sim1} reports the average number of samples at the time of stopping as well as the false-bucket-rate, in the four cases resulting from the combination of the two studied tolerances and the two overlapping bucket configurations. The false-bucket-rate clearly stays very well below the tolerance in all cases. This is not surprising as Corollary~\ref{cor_fbr} in combination with Lemma~\ref{lemma_new_FWER} shows that the tolerance is conservative, in particular if the true value is not close to the split points. Indeed, while still conservative, the \textbf{FBR} is highest for (c) where the true value is closer to a split point than in the other cases. In each case where a decision was made for a false bucket, the false bucket was a direct neighbour of a correct bucket.

Furthermore, we illustrate the obtained results by providing  the boxplots of the stopping times in Figure~\ref{fig_boxplot_ModifiedBandDepth}.  Clearly, more samples are required to make a decision for narrower buckets, i.e.\ scenarios \textbf{W1} as well as a smaller tolerance \textbf{T1}, where the widths of the buckets has a stronger influence than the tolerance. For the different values in (a) to (c) the stopping time depends mainly on the signal-to-noise-ratio, where the signal is given by the distance to the bucket boundary (in case of overlapping ones the maximum of all buckets that contain the true value) and the noise is given by $\V (H(\bm{\xi}))$. In the above example, the empirical variance of all summands for exactly computing the empirical depth can be thought of as a surrogate and is given by 
$ 0.0040$ for (a), $0.0374$ for (b) and by $0.0665$ for (c).

\begin{table}[!b]
        \setlength{\tabcolsep}{1.5mm}
        \begin{center}
  \begin{tabular}{lrr rr rrr rr rrr}
&&\multicolumn{2}{c
}{\text{\textbf{W1T1}}}&& \multicolumn{2}{c 
}{\text{\textbf{W1T2}}}&&\multicolumn{2}{c 
}{\text{\textbf{W2T1}}}&&\multicolumn{2}{c 
}{\text{\textbf{W2T2}}}\\
&& average  & && average  & && average  & && average  & \\
&&samples & \textbf{FBR}&&samples & \textbf{FBR}&&samples & \textbf{FBR}&&samples & \textbf{FBR}\\\hline
(a) $x(t)=2$&&67.224 &0.0 \%&& \phantom{0}61.419 & 0.0 \% && 27.924&0.0 \%&& 26.149 &0.0 \%\\
(b) $x(t)=\sin(2\pi t)$&&1460.702 &0.0 \%&&1181.898&0.1 \%&&417.079&0.0 \%&&367.378&0.0 \%\\
(c) $x(t)=B_t$ &&1732.470&0.1 \%&& 1500.084&0.3 \%&&518.643&0.0 \%&&460.142&0.1 \%\\ 
\end{tabular}  
\end{center}
\caption{Using the modified band depth, average number of samples  taken at time of stopping (average samples) and false-bucket-rate (\textbf{FBR}) for (a) $x(t)=2$, (b) $x(t)=\sin(2\pi t)$ and (c) $x(t)=B_t$ (Brownian motion). \textbf{WiTj} stays for the combination of tolerance \textbf{Ti} with bucket configuration \textbf{Wj}, for $i,j\in\{1,2\}$.} 
\label{table_new_sim1}
\end{table}

\begin{table}[!b]
\setlength{\tabcolsep}{1.5mm}
\begin{center}
\begin{tabular}{lr rr rrrr}
    & &{Exact}&& \textbf{W1T1} & \textbf{W1T2} & \textbf{W2T1} & \textbf{W2T2} \\ \hline
    (a)  $x(t)=2$ && 93.5 && 3.2   &  2.9  & 1.3  & 1.2    \\
    (b) $x(t)=\sin(2\pi t)$ && 93.1 && 72.8  & 59.7  & 21.1    & 17.4    \\
    (c) $x(t)=B_t$ (Brownian motion) && 94.7 && 85.2   & 77.0 &25.3  & 22.8    \\ 
                
\end{tabular}
 \end{center}
\caption{Average runtime in milliseconds (ms) of the exact empirical depth calculation for the modified band depth and the seMCD-algorithms. The computations were run on a computer with an i9-8950HK CPU and 32 GB RAM.} 
\label{tabModFuncBandRuntime}
\end{table} 

\begin{figure}[!t]
\centering
\begin{subfigure}[h]{0.325\columnwidth}
        \includegraphics[width=\columnwidth, trim={3.7cm 9.5cm 4.3cm 10cm},clip]{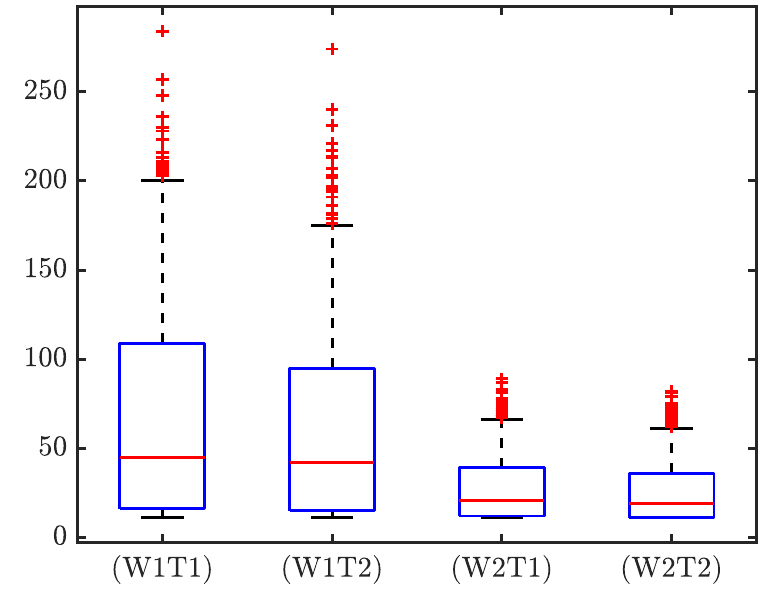}
        \vspace*{-4ex}
        \caption{\centering $x(t) = 2$.}
\end{subfigure}
\begin{subfigure}[h]{0.325\columnwidth}
        \includegraphics[width=\columnwidth, trim={3.7cm 9.5cm 4.3cm 10cm},clip]{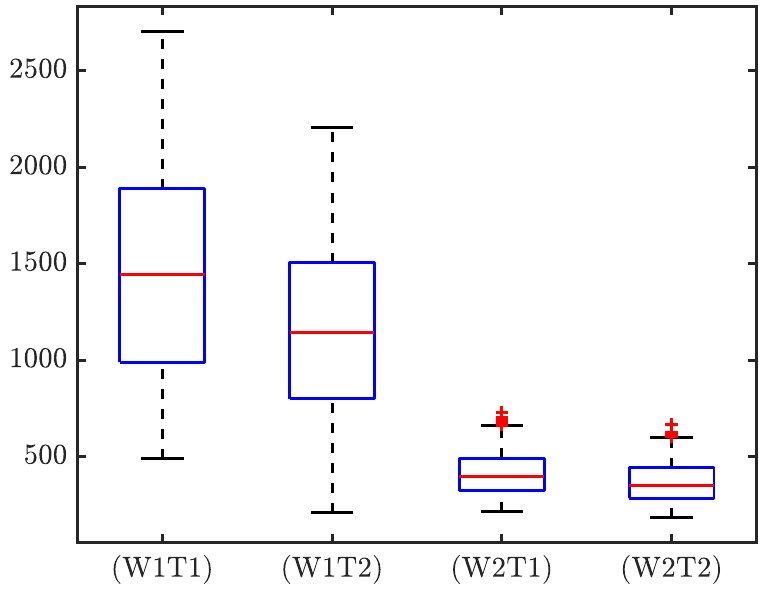}
        \vspace*{-4ex}
        \caption{\centering $x(t) = \sin(2 \pi t)$.}
\end{subfigure}
\begin{subfigure}[h]{0.325\columnwidth}
        \includegraphics[width=\columnwidth, trim={3.7cm 9.5cm 4.3cm 10cm},clip]{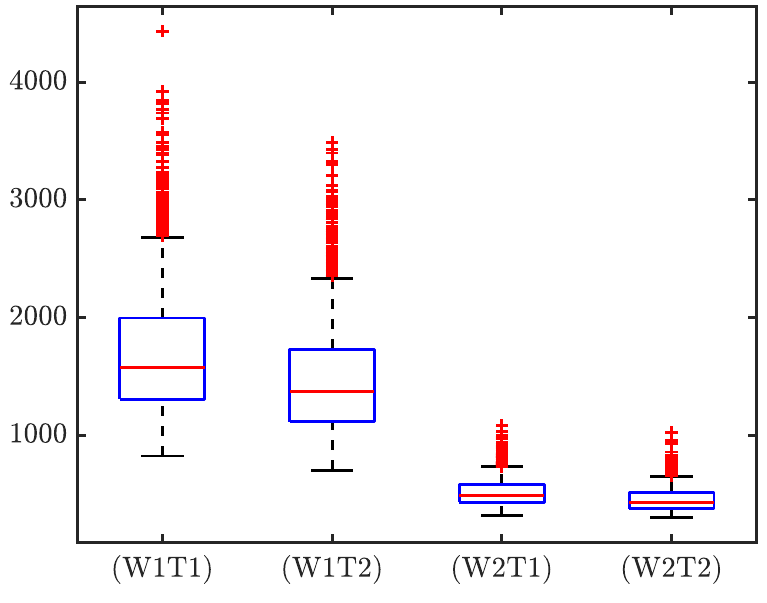}
        \vspace*{-4ex}
        \caption{\centering $x(t) = B_t$ (Brownian motion).}
\end{subfigure}
\caption{Boxplots for the number of samples at time of stopping for the modified band depth.}
\label{fig_boxplot_ModifiedBandDepth}
\end{figure}  

\begin{figure}[!t]
\centering
\begin{subfigure}[h]{0.325\columnwidth}
        \includegraphics[width=\columnwidth, trim={3.7cm 9.5cm 4.3cm 10cm},clip]{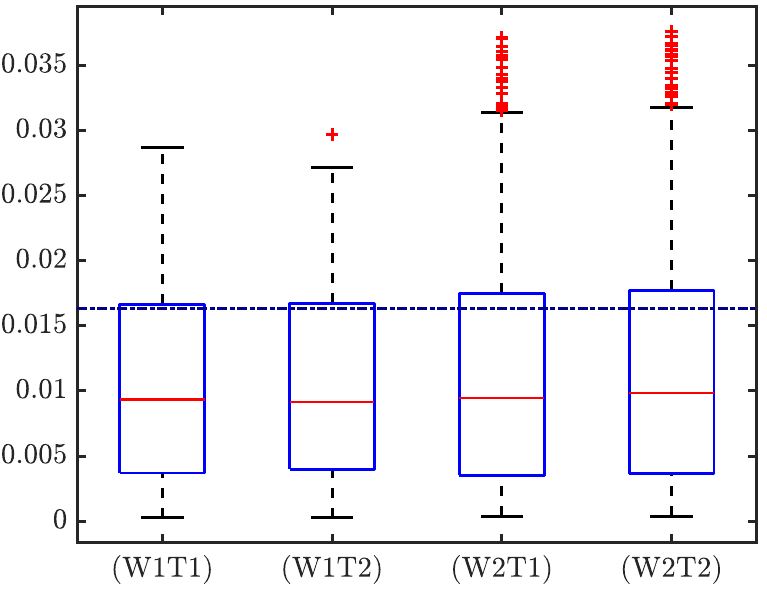}
        \vspace*{-4ex}
        \caption{\centering $x(t) = 2$.}
\end{subfigure}
\begin{subfigure}[h]{0.325\columnwidth}
        \includegraphics[width=\columnwidth, trim={3.7cm 9.5cm 4.3cm 10cm},clip]{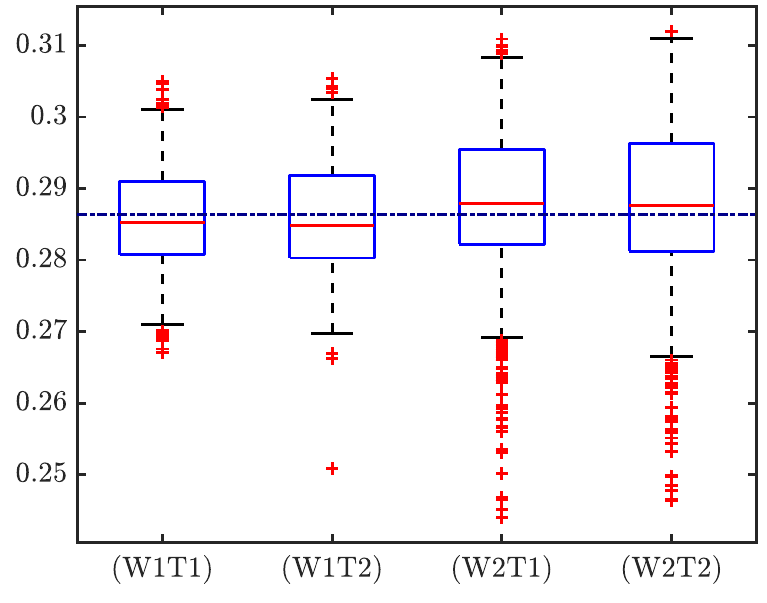}
        \vspace*{-4ex}
        \caption{\centering $x(t) = \sin(2 \pi t)$.}
\end{subfigure}
\begin{subfigure}[h]{0.325\columnwidth}
        \includegraphics[width=\columnwidth, trim={3.7cm 9.5cm 4.3cm 10cm},clip]{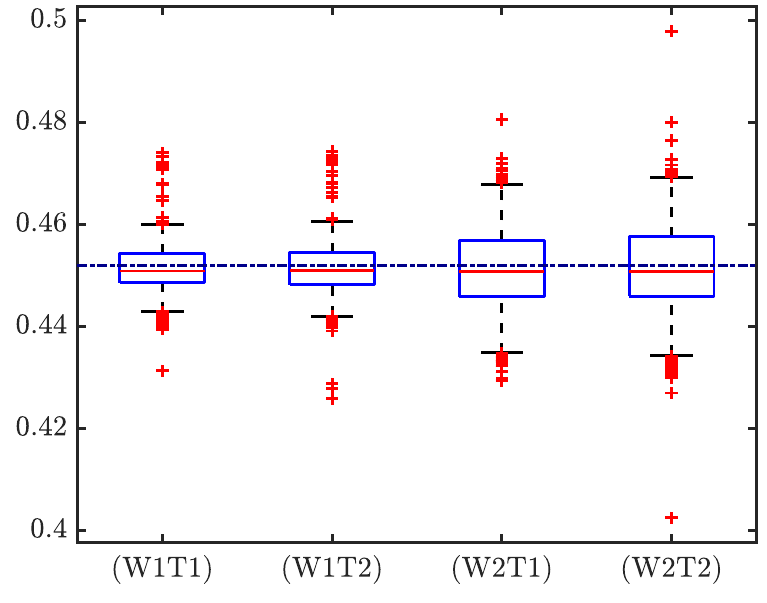}
        \vspace*{-4ex}
        \caption{\centering $x(t) = B_t$ (Brownian motion).}
\end{subfigure}
\caption{Boxplots for the point estimates $\hat{h}_{\tau}$ for the modified band depth. The dash-dotted line in the boxplots indicated the exact value. The plot shows that these values have to be interpreted with care.}
\label{fig_boxplot_ModifiedBandDepth1}
\end{figure}     

The number of required samples gives important information on computational aspects of the algorithm with the advantage that it does not depend on the implementation or the computer used. Having said that, we complement this information in Table~\ref{tabModFuncBandRuntime} with the average runtime in milliseconds. In each of the three cases, the seMCD-algorithm is faster than the exact calculation.
    
By construction of the non-parametric boundary sequences in Section~\ref{sec_sequential_construction}, it is guaranteed that the point estimate at the time of stopping, i.e.\ $\hat{h}(\tau)=\frac{1}{\tau}\sum_{j=1}^{\tau}H(\bm{\xi}_j)$ with $\tau=\tau_m(0.1,0.4)$ as in Corollary~\ref{cor_47} is inside the output seMCD-bucket. However, the precision of $\hat{h}(\tau)$ depends on the random stopping time $\tau$ and is particularly affected by the distance of the true value to the split points, where less samples are required to decide with some statistical certainty, that the true parameter is inside the bucket, if this distance is larger. Furthermore, due to the random $\tau$ depending on the samples taken, $\hat{h}(\tau)$ does not inherit the statistical properties of the sample mean for fixed sample size such as e.g.\ unbiasedness. We illustrate the behaviour of $\hat{h}(\tau)$ by the boxplots in Figure~\ref{fig_boxplot_ModifiedBandDepth1}. Clearly, the estimator in (a) is strongly biased, unlike the sample mean with fixed sample size, which can be explained by the fact that the true value is only contained in the lowest bucket. Therefore, the smaller the partial sum process happens to be, the quicker a rejection will occur. This is because the only relevant boundary sequence for the lowest bucket in the overlapping case is the lower boundary sequence of the split point $h_2$. In the other cases, this asymmetry is less pronounced if somewhat visible in (b), where it depends on the unbalancedness of the respective distance of the true value to both bucket boundary values of the true seMCD-bucket(s). Less surprisingly, the estimator becomes more precise with smaller width of the buckets (\textbf{W1}) and less pronounced also with smaller tolerance (\textbf{T1}) due to $\tau$ tending to be larger in these cases.

In conclusion, one has to be very careful with the interpretation of $\hat{h}(\tau)$. 

\subsubsection{IRW depth}
While the empirical depth in Section~\ref{subsec_sim_mod_band_depth} allows for an exact computation, this is not the case for the multivariate IRW depth (see Appendix~\ref{appendix_depths}), which requires Monte Carlo methods for its calculation. In fact, \cite{Ramsay2019} propose to use a plain-vanilla Monte Carlo approach based on $N=10^5$ samples, while \cite{Staerman2021} propose to use $N=100\cdot d$ samples. For dimension $d=100$, we  compare these values of $N$ with the number of samples at the stopping time of the seMCD-algorithms, in the four combinations of the tolerances and bucket configurations provided at the beginning of Section~\ref{fs}. To this end, we take a fixed sample of 100 iid $d$-variate standard Gaussian observations. Then, we apply the different Monte Carlo techniques $1000$ times to compute the empirical depth of the vector $ e \cdot (1, \ldots,1)^T \in \mathbb R^d$, for (a) $e=0$, (b) $e=0.5$ and (c) $e=5$, with respect to the fixed sample above. 

\begin{figure}[htb]
        \centering
        \begin{subfigure}[h]{0.325\columnwidth}
        \includegraphics[width=\columnwidth, trim={3.7cm 9.5cm 4.3cm 10cm},clip]{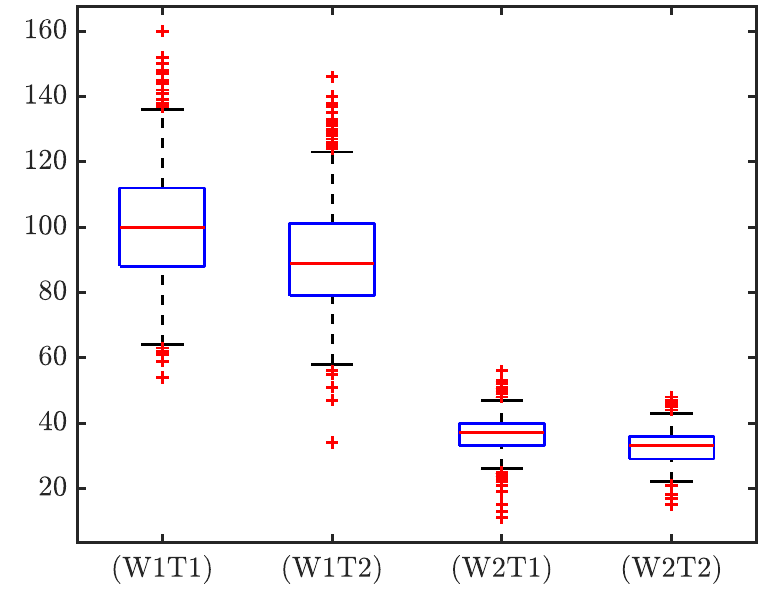}
        \vspace*{-4ex}
        \caption{\centering $e=0$.}
    \end{subfigure}
    \begin{subfigure}[h]{0.325\columnwidth}
        \includegraphics[width=\columnwidth, trim={3.7cm 9.5cm 4.3cm 10cm},clip]{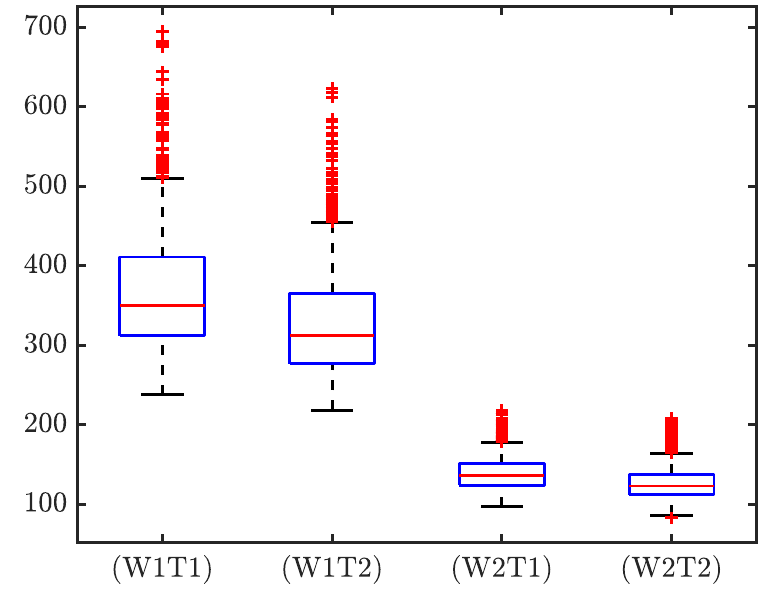}
        \vspace*{-4ex}
        \caption{\centering $e=0.5$.}
    \end{subfigure}
    \begin{subfigure}[h]{0.325\columnwidth}
                \includegraphics[width=\columnwidth, trim={3.7cm 9.5cm 4.3cm 10cm},clip]{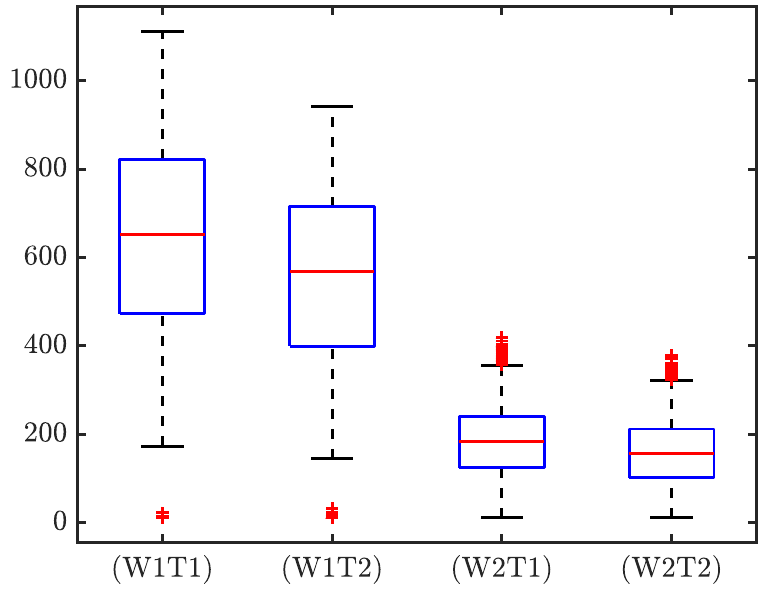}
        \vspace*{-4ex}
        \caption{\centering $e=5$.}
    \end{subfigure}
    \caption{Boxplots for the number of samples taken at time of stopping for the seMCD-algorithms for the IRW depth in dimension $d=100$.}
    \label{fig_boxplot_IRW}
\end{figure} 
    
\begin{table}[!t]
            \setlength{\tabcolsep}{1.5mm}
        \begin{center}
            \begin{tabular}{lr r rr rrrr}
           && \multicolumn{2}{c
           }{Plain-Vanilla}&&\multicolumn{4}{c
           }{seMCD-algorithms}\\[1mm]
                && {$N=10^5$}& $N=100\cdot d$&& \textbf{W1T1} & \textbf{W1T2} & \textbf{W2T1} & \textbf{W2T2} \\ \hline
         (a)  $e=0$ &&  1\,643.9  & 265.0&& 22.5    & 19.5 & 7.1 &  7.0   \\  
            (b)  $e=0.5$ &&   1\,709.6  & 251.9&& 89.0  & 79.1  & 29.5    & 28.5\\
            (c)  $e=5$&&  1\,761.2 & 244.8 && 154.7   & 133.2 & 41.7 & 35.9 \\
            \end{tabular}
            \end{center}
            \caption{Average runtime in milliseconds (ms) for the IRW depth of the vector $e\cdot (1,\ldots,1)^T$ with $d=100$ (based on 1000 trials). The computations were run on a computer with an i9-8950HK CPU and 32 GB RAM.} 
            \label{tabIRWRuntime}
\end{table}

Figure~\ref{fig_boxplot_IRW} displays the stopping times for the seMCD-algorithms considered.  Clearly, the average number of samples required has a magnitude smaller than the $10\,000$ proposed by \cite{Staerman2021} for $d=100$ and, consequently, of the $10^5$ proposed in \cite{Ramsay2019}. In fact, with the exception of very few runs, in the case of \textbf{W1T1} with $e=5$, the stopping times are smaller than $1000$; with the majority much smaller. Additionally, Table~\ref{tabIRWRuntime} complements this information by reporting the average runtimes, in the different studied cases, of the seMCD-algorithm and the proposals in  \cite{Ramsay2019} and \cite{Staerman2021}. The obtained runtimes are much smaller for the seMCD-algorithm, which confirms the previously made observations.

\subsection{Anomaly detection for smtp data}\label{ss}  
We consider data from the 1999 KDD Cup Network Intrusion Detection Competition \citep{KDDCup1999}, where the task was to distinguish cyberattacks from normal network connections simulated in a military network environment. In particular, we focus on the smtp data set\footnote{\url{https://maxhalford.github.io/files/datasets/smtp.zip}, Retrieved on April 29, 2025.},  a data set consisting of 95\,156 observations in $\mathbb R^3$. This data set has previously been considered in \citet{Williams2002,Yamanishi2004,Liu2008,Ting2010,Tan2011} and in \cite{mozharovskyi2024anomalydetectionusingdata}. The latter study explored data depth as a tool for anomaly detection and benchmarked its performance against several other methods, including non-depth-based approaches. Due to the considerable size of the data set, a subset of 10\,000 points was used for comparison. However, this reduced sample may not sufficiently capture the complexity of the data structure, as none of the methods evaluated in \cite{mozharovskyi2024anomalydetectionusingdata} succeeded in identifying a true anomaly. Thus, we use the entire data set in this study, which is shown in Figure~\ref{fig_anomaly_detection}.  Panel~(a) highlights 30 points in red that correspond to the cyberattacks, which we aim to detect as anomalies. These points are 0.0315 \% of the data. 

\begin{figure}[htb]
        \centering
        \begin{subfigure}[h]{0.49\columnwidth}
        \includegraphics[width=\columnwidth,trim={3.85cm 0cm 3cm 0cm}, clip]{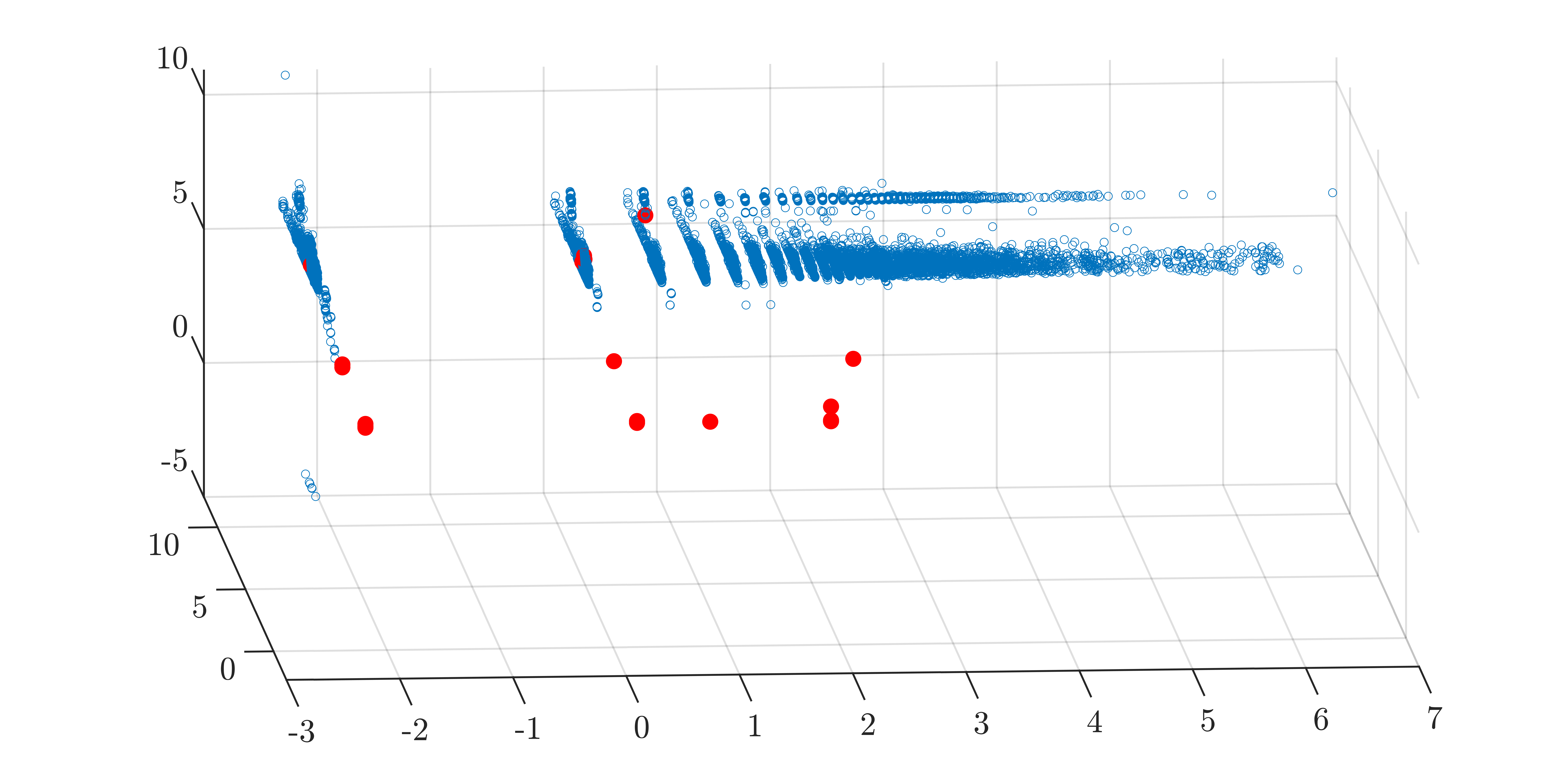}
        \vspace*{-4ex}
        \caption{\centering The points highlighted in red are the original anomalies.}
        \captionsetup{justification=centering}
    \end{subfigure}
    \begin{subfigure}[h]{0.49\columnwidth}
        \includegraphics[width=\columnwidth,trim={3.85cm 0cm 3cm 0cm}, clip]{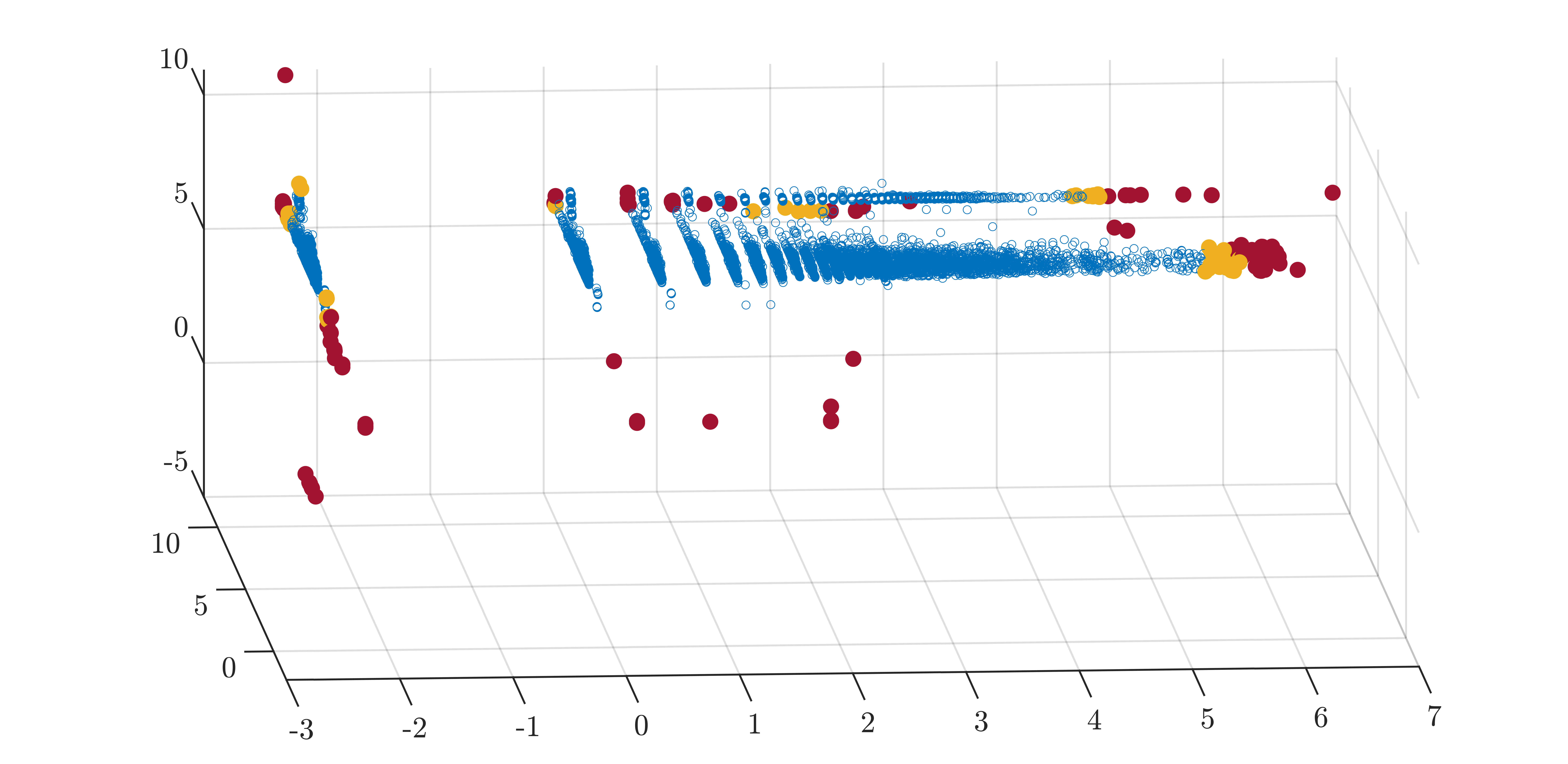}
        \vspace*{-4ex}
        \caption{\centering Points classified as anomalies (burgundy) and undecided points (orange).}
    \end{subfigure}
    \caption{Three-dimensional scatter plots of the smtp data set.}
    \label{fig_anomaly_detection}
    \end{figure}

In the seMCD-algorithm, we use the spherical depth \citep{elmore2006} with the overlapping buckets $(0,10^{-3}),(5^{-4},15^{-4}),(10^{-3},1)$. As the spherical depth is Type A with indicator kernel, we use the boundary sequences from Section~\ref{sec_gandy_approach} with tolerance $\alpha=0.01$ and the default spending sequence from \eqref{eq_default_spending}.

\begin{table}[htb]
 \begin{center}
\begin{tabular}{c c c c c c c c }
&\multicolumn{6}{c}{ Quantiles} &\\
minimum & 25 \% & 50 \% & 75 \% & 90 \% & 95 \% & 99 \% & maximum\\\hline
2& 4&7&17&46&93&977.94&378\,702\\ 
    \end{tabular}
    \end{center}
    \caption{Number of seMCD-samples taken before the seMCD-algorithm stopped when applied to each data point.}
    \label{table_quantiles_anomalies}
 \end{table}

The seMCD-algorithm outputs the bucket:
\begin{itemize}
\item   $(0,10^{-3})$ for 127 points (0.1335 \% of the data); including 20 cyberattacks.
\item  $(10^{-3},1)$ for 94\,950 points (99.7835 \% of the data); including 10 cyberattacks.
\item $(5^{-4},15^{-4})$ for 79 points (0.0830 \% of the data); none of them is a cyberattack.
\end{itemize}
Thus, the seMCD-algorithm detects 127 points as anomalies. As the possible depth values range in the interval (0,1), a split point of $10^{-3}$ can be considered to correspond to 0.1 \% of the data. Thus, it is not surprising that 0.1335 \% of the data are detected as an anomaly. The 10 cyberattacks  that are not detected as anomalies by the seMCD-algorithm are located within the layers of blue points and have  depth values higher than the $10^{-3}$ split point; see Figure~\ref{fig_anomaly_detection}. Thus, even with exact computation it would not be possible to detect these anomalies with the spherical depth.  Furthermore, the points classified as anomalies that are not among the original anomalies clearly are located in the outer layers of the blue points with the undecided points often in between the points not classified as anomalies by the algorithm (in blue) and the points classified as anomalies (in burgundy), in Panel~(b) of Figure~\ref{fig_anomaly_detection}. In consequence, the seMCD-algorithm behaves as expected and the points classified as anomalies, that were not among the original anomalies, are anomalies in the sense implied by the spherical depth but not with respect to the cyberattack. 

The median number of seMCD-samples at the time of stopping is 7 but the distribution is highly skewed as can be seen by the other reported quantiles as in Table~\ref{table_quantiles_anomalies}. As a consequence of this skewness the mean number of seMCD-samples taken before the algorithm stops is 308.837.

Thus, determining the seMCD-bucket of a point in the smtp data set usually requires only a few samples, but in some cases it may take much more. 

\subsection{Binary classification}\label{subsec_binary_classification}
In this section, we illustrate the performance of the seMCD-algorithm for maximum depth classification both with simulated data and based on a data example.
  
\subsubsection{Simulation study}
To illustrate the performance of the seMCD-algorithm using simulated data, we use the modified band depth.  This choice allows us to compare the seMCD results with the true (empirical) maximum depth classifier. 

We use a fixed training data set consisting of $100$ observations from a standard Brownian motion (Class~1) and $100$ observations from a standard Brownian motion shifted by $+1$ (Class 2). The test data set contains $1000$ observations: $500$ from each of the two models described in the cases below. The percentages in parentheses indicate the proportion of observations from each model that were assigned to Class 1 by the maximum depth classifier using the exact empirical depth values.
\pagebreak
\begin{enumerate}[label=(\Roman*)]
\item (i) $(B_t)_{t \in [0,1]}$ (88.6 \%) and (ii) $(B_t+1)_{t \in [0,1]}$ (16.8 \%).
\item (i) $(B_t+0.1)_{t \in [0,1]}$ (78.8 \%) and (ii) $(B_t+0.9)_{t \in [0,1]}$ (25.0 \%).\label{TestData2}
\item (i) $(B_t+0.2)_{t \in [0,1]}$ (74.0 \%) and (ii) $(B_t+0.8)_{t \in [0,1]}$ (31.2 \%).
\item (i) $(B_t+0.3)_{t \in [0,1]}$ (65.0 \%) and (ii) $(B_t+0.7)_{t \in [0,1]}$ (33.8 \%).\label{TestData4}
\end{enumerate}

\begin{table}[!t]
        \setlength{\tabcolsep}{1.5mm}
        \begin{center}
            \begin{tabular}{l rr rr rr rr} 
               &\multicolumn{2}{c}{I}&\multicolumn{2}{c} {II}&\multicolumn{2}{c} {III}&\multicolumn{2}{c} {IV}\\
                & 
                \multicolumn{1}{c}{(i)} &   \multicolumn{1}{c}{(ii)} & 
                 \multicolumn{1}{c}{(i)} & \multicolumn{1}{c}{(ii)} &  
                 \multicolumn{1}{c}{(i)} & \multicolumn{1}{c}{(ii)} & 
                 \multicolumn{1}{c}{(i)} & \multicolumn{1}{c}{(ii)}\\\hline 
                Decision for class 1&85.6 \%&16.0 \%&77.8 \%&23.4 \%&73.0 \%&29.2 \%&63.6 \%& 31.4 \%\\
                Decision for class 2&10.8 \%&82.8 \%&18.6 \%&73.6 \%&24.6 \%&67.4 \%&34.0 \%&63.2 \%\\
                Decision for $(-0.02,0.02)$&1.2 \%&0.2 \%&1.2 \%&1.2 \%&0.8 \%&1.0 \%&1.2 \%&1.8 \%\\
                Not stopped &2.4 \%&1.0 \%&2.4 \%&1.8 \%&1.6 \%& 2.4 \%&1.2 \%&3.6 \%\\
            \end{tabular}
            \end{center}
            \caption{Results of the test data scenarios with maximum depth classification via seMCD-buckets.} 
            \label{tabBucketClassificationNEW}
\end{table}  

Clearly, the distributions of the test data sets \ref{TestData2}-\ref{TestData4} are different from the distributions of the training data set.

In the seMCD-algorithm, we use a tolerance of $\alpha=1 \ \%$ and the overlapping buckets $[-1,0)$, $(0,1]$, $(-\epsilon, \epsilon)$, with  $\epsilon=0.02$. The results are displayed in Table \ref{tabBucketClassificationNEW}. All decisions made for observations falling into the $[-1,0)$ and $(0,1]$ buckets coincide with those made by the maximum depth classifier using the exact depth values. Monitoring is stopped after a maximum of $4950$ samples, which corresponds to the number of summands in the exact depth calculation. At this point, a few observations are assigned to the bucket $(-\epsilon, \epsilon)$, and some have not yet reached a decision. A closer inspection shows that those in the bucket $(-\epsilon, \epsilon)$ typically have true depth value very close to zero and that while the undecided ones also have a small true depth value, it is closer to $0.02$.

\subsubsection{Banana data example}
We consider the banana data set \citep{BananaData}, which contains the data of $8000$ bananas and their class, given by a binary value that determines the quality of the banana. $4006$ bananas have good quality. The data set contains different features of the bananas, among which we use here the size, weight, softness, and ripeness. Figure \ref{fig_boxplots_banana} contains a boxplot  for each of these features.  

We randomly pick $1000$ bananas from each class as the training data set. We use the remaining $6000$ bananas as the test data set. For the seMCD-algorithm, we use the spherical depth (Type A with indicator kernel as in \eqref{he}), which has the advantage that it does not contract to $0$ when $d$ grows. See Remark~\ref{rem_simplicial} for a depth that does contract. We use a tolerance of $\alpha=1 \ \%$ and the overlapping buckets $[-1,0)$, $(0,1]$, $(-\epsilon, \epsilon)$ with $\epsilon=0.001$. For each datum, if the algorithm returns the bucket $[-1,0)$, the datum will be assigned to one class and, if it returns the bucket $(0,1]$, it will be assigned to the other. It will not be assigned to any class if it returns the bucket $(-\epsilon, \epsilon)$.
Additionally, if the seMCD-algorithm has not returned a bucket after
 $4.5 \cdot 10^5$ samples, we stop the algorithm without a decision. This upper bound is smaller than the number of summands that have to be evaluated for an exact computation, which is given by ${1000} \choose 2$ for each of the two training data sets.

\begin{figure}[htb]
\centering
\includegraphics[width=0.65\columnwidth, trim={2cm 2.5cm 5cm 2.1cm},clip]{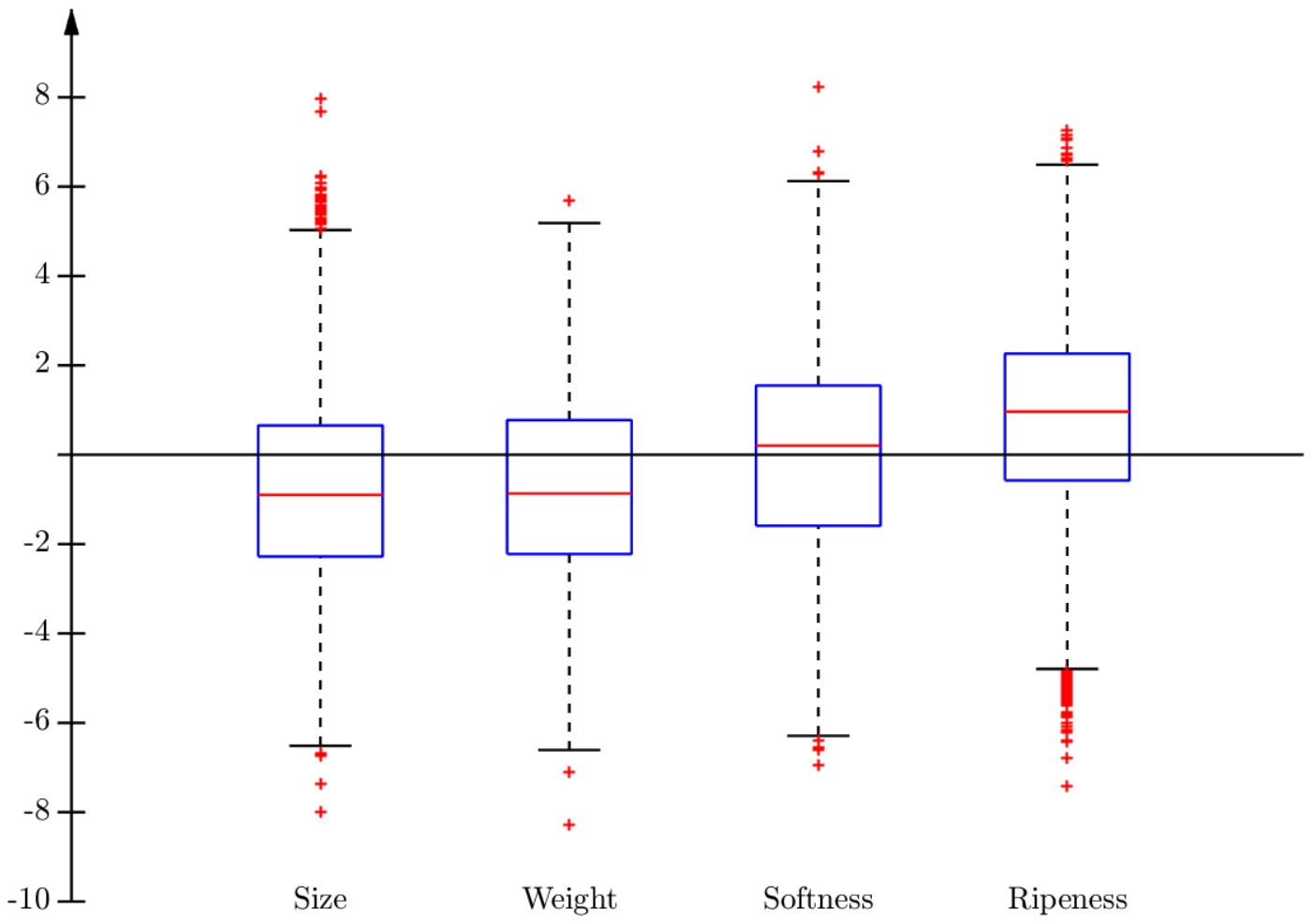}
\caption{Boxplots for each considered feature in the banana data set.}
\label{fig_boxplots_banana}
\end{figure} 

The bucket-based algorithm: 
    \begin{itemize}
        \item[--] Assigns 83.52 \% of the test data to the correct class.
        \item[--] Assigns 14.72 \%  of the test data to the wrong class.
        \item[--] Returns the bucket $(-\epsilon, \epsilon)$ for 0.50 \% of the test data. These data are not assigned to either of the two classes. 
        \item[--] Is stopped for 1.27 \% of the test data after $4.5 \cdot 10^5$ sampling steps. These data are not assigned to either of the two classes.
    \end{itemize}
   Despite the approximately 2 \% of cases for which no decision was reached, the seMCD-based classifier using spherical depth slightly outperforms MATLAB’s naïve Bayes 
classifier,  which correctly classifies only 
    82.52 \% of the test data.
    
\section{Discussion}
In this paper, we propose and explain the seMCD-algorithm designed to compute buckets that contain the true depth-based quantity of interest. This includes applications, such as empirical depth computation and maximum depth classification. The methodology is suitable for a broad class of depth functions, both multivariate and functional, that includes all Type A depth functions. 
    
For the special case of Type A depth functions with indicator kernel, we adopt boundary sequences originally introduced by ~\cite{Gandy2009} in another context. Adapting this methodology yields exact statistical guarantees. We provide a detailed description of the algorithm to compute these boundary sequences and analyse the implications for the statistical guarantees that stem from their construction. 

Since this exact methodology is limited to Type A depth functions with indicator kernel, we also propose a general methodology that works under suitable second moment, even if the underlying distribution of the summands in the seMCD-process is unknown. This generality comes at the cost of replacing the exact guarantees with asymptotic ones. Finally, we show the good performance of the seMCD-methodology in a simulation study and in a real data analysis.

\section*{Acknowledgements and funding}
This work was financially supported by the Deutsche Forschungsgemeinschaft (DFG, German Research Foundation) - 314838170, GRK 2297 MathCoRe and by grant PID2022-139237NB-I00 funded by “ERDF A way of making Europe” and MCIN/AEI/10.13039/501100011033. This article builds on Chapter II of the PhD thesis of the first author \citep{GnettnerDissertation2024}.

\bibliography{literatureCOMPLETE}

\begin{thebibliography}{61}
\providecommand{\natexlab}[1]{#1}
\providecommand{\url}[1]{\texttt{#1}}
\expandafter\ifx\csname urlstyle\endcsname\relax
  \providecommand{\doi}[1]{doi: #1}\else
  \providecommand{\doi}{doi: \begingroup \urlstyle{rm}\Url}\fi

\bibitem[KDD(1999)]{KDDCup1999}
Kdd {C}up 1999 {D}ata, 1999.
\newblock URL \url{https://kdd.ics.uci.edu/databases/kddcup99/kddcup99.html}.
\newblock Retrieved April 29, 2025.

\bibitem[Afshani et~al.(2016)Afshani, Sheehy, and Stein]{Afshani2016}
P.~Afshani, D.~R. Sheehy, and Y.~Stein.
\newblock Approximating the simplicial depth in high dimensions.
\newblock In \emph{The European Workshop on Computational Geometry}, 2016.

\bibitem[Aloupis et~al.(2002)Aloupis, Cortés, Gómez, Soss, and
  Toussaint]{Aloupis2002}
G.~Aloupis, C.~Cortés, F.~Gómez, M.~Soss, and G.~Toussaint.
\newblock Lower bounds for computing statistical depth.
\newblock \emph{Computational Statistics \& Data Analysis}, 40\penalty0
  (2):\penalty0 223--229, 2002.
\newblock ISSN 0167-9473.
\newblock \doi{10.1016/S0167-9473(02)00032-4}.

\bibitem[Arribas-Gil and Romo(2014)]{Arribas}
A.~Arribas-Gil and J.~Romo.
\newblock {Shape outlier detection and visualization for functional data: the
  outliergram}.
\newblock \emph{Biostatistics}, 15\penalty0 (4):\penalty0 603--619, 03 2014.
\newblock ISSN 1465-4644.
\newblock \doi{10.1093/biostatistics/kxu006}.

\bibitem[Aue and Kirch(2023)]{ClaudiaAlex2023}
A.~Aue and C.~Kirch.
\newblock {The state of cumulative sum sequential changepoint testing 70 years
  after Page}.
\newblock \emph{Biometrika}, page asad079, 12 2023.
\newblock ISSN 1464-3510.
\newblock \doi{10.1093/biomet/asad079}.

\bibitem[Baharav and Lai(2023)]{BaharavLai2023}
T.~Baharav and T.~L. Lai.
\newblock Adaptive data depth via multi-armed bandits.
\newblock \emph{Journal of Machine Learning Research}, 24\penalty0
  (155):\penalty0 1--29, 2023.

\bibitem[Briend et~al.(2023)Briend, Lugosi, and Oliveira]{Briend2023}
S.~Briend, G.~Lugosi, and R.~I. Oliveira.
\newblock On the quality of randomized approximations of {T}ukey's depth.
\newblock \emph{arXiv preprint arXiv:2309.05657v2}, 2023.
\newblock \doi{10.48550/arXiv.2309.05657}.

\bibitem[Cascos and Molchanov(2007)]{Cascos2007}
I.~Cascos and I.~Molchanov.
\newblock Multivariate risks and depth-trimmed regions.
\newblock \emph{Finance and Stochastics}, 11\penalty0 (3):\penalty0 373--397,
  Jul 2007.
\newblock ISSN 1432-1122.
\newblock \doi{10.1007/s00780-007-0043-7}.

\bibitem[Chakraborty and Chaudhuri(2014)]{SpatialDepth}
A.~Chakraborty and P.~Chaudhuri.
\newblock {The spatial distribution in infinite dimensional spaces and related
  quantiles and depths}.
\newblock \emph{The Annals of Statistics}, 42\penalty0 (3):\penalty0 1203 --
  1231, 2014.
\newblock \doi{10.1214/14-AOS1226}.

\bibitem[Chen et~al.(2009)Chen, Dang, Peng, and Bart]{ChenDang2010}
Y.~Chen, X.~Dang, H.~Peng, and H.~Bart.
\newblock Outlier detection with the kernelized spatial depth function.
\newblock \emph{IEEE Transactions on Pattern Analysis and Machine
  Intelligence}, 31\penalty0 (2):\penalty0 288--305, 2009.
\newblock \doi{10.1109/TPAMI.2008.72}.

\bibitem[Cheng and Ouyang(2001)]{Cheng2001}
A.~Y. Cheng and M.~Ouyang.
\newblock On algorithms for simplicial depth.
\newblock In \emph{Proceedings of the 13th Canadian Conference on Computational
  Geometry (CCCG'01)}, pages 53--56, 2001.

\bibitem[Chu et~al.(1996)Chu, Stinchcombe, and White]{Chu1996}
C.-S.~J. Chu, M.~Stinchcombe, and H.~White.
\newblock Monitoring structural change.
\newblock \emph{Econometrica}, 64\penalty0 (5):\penalty0 1045--1065, 1996.
\newblock ISSN 00129682, 14680262.
\newblock \doi{10.2307/2171955}.

\bibitem[Cl{\'e}men{\c{c}}on et~al.(2023)Cl{\'e}men{\c{c}}on, Mozharovskyi, and
  Staerman]{Staerman2021}
S.~Cl{\'e}men{\c{c}}on, P.~Mozharovskyi, and G.~Staerman.
\newblock {Affine invariant integrated rank-weighted statistical depth:
  properties and finite sample analysis}.
\newblock \emph{Electronic Journal of Statistics}, 17\penalty0 (2):\penalty0
  3854 -- 3892, 2023.
\newblock \doi{10.1214/23-EJS2189}.

\bibitem[Cuesta-Albertos and Nieto-Reyes(2008)]{AliciaCuesta}
J.~Cuesta-Albertos and A.~Nieto-Reyes.
\newblock The random {T}ukey depth.
\newblock \emph{Computational Statistics \& Data Analysis}, 52\penalty0
  (11):\penalty0 4979--4988, 2008.
\newblock ISSN 0167-9473.
\newblock \doi{10.1016/j.csda.2008.04.021}.

\bibitem[Cuevas and Fraiman(2009)]{CuevasFraimanIntegrated2009}
A.~Cuevas and R.~Fraiman.
\newblock On depth measures and dual statistics. a methodology for dealing with
  general data.
\newblock \emph{Journal of Multivariate Analysis}, 100\penalty0 (4):\penalty0
  753--766, 2009.
\newblock ISSN 0047-259X.
\newblock \doi{10.1016/j.jmva.2008.08.002}.

\bibitem[Cuevas et~al.(2007)Cuevas, Febrero, and Fraiman]{Cuevas2007Projection}
A.~Cuevas, M.~Febrero, and R.~Fraiman.
\newblock Robust estimation and classification for functional data via
  projection-based depth notions.
\newblock \emph{Computational Statistics}, 22\penalty0 (3):\penalty0 481--496,
  Sep 2007.
\newblock ISSN 1613-9658.
\newblock \doi{10.1007/s00180-007-0053-0}.

\bibitem[Dang and Serfling(2010)]{DangSerfling2010}
X.~Dang and R.~Serfling.
\newblock Nonparametric depth-based multivariate outlier identifiers, and
  masking robustness properties.
\newblock \emph{Journal of Statistical Planning and Inference}, 140\penalty0
  (1):\penalty0 198--213, 2010.
\newblock ISSN 0378-3758.
\newblock \doi{10.1016/j.jspi.2009.07.004}.

\bibitem[Ding et~al.(2020)Ding, Gandy, and Hahn]{Gandy2020a}
D.~Ding, A.~Gandy, and G.~Hahn.
\newblock A simple method for implementing monte carlo tests.
\newblock \emph{Computational Statistics}, 35\penalty0 (3):\penalty0
  1373--1392, Sep 2020.
\newblock ISSN 1613-9658.
\newblock \doi{10.1007/s00180-019-00927-6}.

\bibitem[Dyckerhoff and Mozharovskyi(2016)]{Dyckerhoff1}
R.~Dyckerhoff and P.~Mozharovskyi.
\newblock Exact computation of the halfspace depth.
\newblock \emph{Computational Statistics \& Data Analysis}, 98:\penalty0
  19--30, 2016.
\newblock ISSN 0167-9473.
\newblock \doi{10.1016/j.csda.2015.12.011}.

\bibitem[Elmore et~al.(2006)Elmore, Hettmansperger, and Xuan]{elmore2006}
R.~T. Elmore, T.~P. Hettmansperger, and F.~Xuan.
\newblock Spherical data depth and a multivariate median.
\newblock In \emph{Data Depth: Robust Multivariate Analysis, Computational
  Geometry and Applications}, volume~72 of \emph{DIMACS Series in Discrete
  Mathematics and Theoretical Computer Science}, pages 87--102. American
  Mathematical Society, 2006.

\bibitem[Febrero et~al.(2008)Febrero, Galeano, and
  González-Manteiga]{FebreroGaleano2008}
M.~Febrero, P.~Galeano, and W.~González-Manteiga.
\newblock Outlier detection in functional data by depth measures, with
  application to identify abnormal nox levels.
\newblock \emph{Environmetrics}, 19\penalty0 (4):\penalty0 331--345, 2008.
\newblock \doi{10.1002/env.878}.

\bibitem[Fischer and Ramdas(2025)]{Fischer2024B}
L.~Fischer and A.~Ramdas.
\newblock Sequential {M}onte {C}arlo testing by betting.
\newblock \emph{Journal of the Royal Statistical Society Series B: Statistical
  Methodology}, page qkaf014, 04 2025.
\newblock ISSN 1369-7412.
\newblock \doi{10.1093/jrsssb/qkaf014}.

\bibitem[Franke et~al.(2022)Franke, Hefter, Herzwurm, Ritter, and
  Schwaar]{franke2022adaptive}
J.~Franke, M.~Hefter, A.~Herzwurm, K.~Ritter, and S.~Schwaar.
\newblock Adaptive quantile computation for {B}rownian bridge in change-point
  analysis.
\newblock \emph{Computational Statistics \& Data Analysis}, 167:\penalty0
  107375, 2022.
\newblock ISSN 0167-9473.
\newblock \doi{10.1016/j.csda.2021.107375}.

\bibitem[Gandy(2009)]{Gandy2009}
A.~Gandy.
\newblock Sequential implementation of {M}onte {C}arlo tests with uniformly
  bounded resampling risk.
\newblock \emph{Journal of the American Statistical Association}, 104\penalty0
  (488):\penalty0 1504--1511, 2009.
\newblock \doi{10.1198/jasa.2009.tm08368}.

\bibitem[Gandy et~al.(2020)Gandy, Hahn, and Ding]{Gandy2020b}
A.~Gandy, G.~Hahn, and D.~Ding.
\newblock Implementing {M}onte {C}arlo tests with p-value buckets.
\newblock \emph{Scandinavian Journal of Statistics}, 47\penalty0 (3):\penalty0
  950--967, 2020.
\newblock \doi{10.1111/sjos.12434}.

\bibitem[Ghosh and Chaudhuri(2005)]{GhoshChaudhuri2005}
A.~K. Ghosh and P.~Chaudhuri.
\newblock On maximum depth and related classifiers.
\newblock \emph{Scandinavian Journal of Statistics}, 32\penalty0 (2):\penalty0
  327--350, 2005.
\newblock \doi{10.1111/j.1467-9469.2005.00423.x}.

\bibitem[Gnettner(2024)]{GnettnerDissertation2024}
F.~Gnettner.
\newblock \emph{Depth functions for multivariate and functional data:
  computation and statistical inference}.
\newblock PhD thesis, Otto-von-Guericke-Universität Magdeburg, 2024.

\bibitem[Gnettner and Kirch(2025)]{gnettner2025newflexibleclasssharp}
F.~Gnettner and C.~Kirch.
\newblock A new and flexible class of sharp asymptotic time-uniform confidence
  sequences.
\newblock \emph{Statistics \& Probability Letters}, 226:\penalty0 110462, 2025.
\newblock ISSN 0167-7152.
\newblock \doi{10.1016/j.spl.2025.110462}.

\bibitem[{González-De La Fuente} et~al.(2022){González-De La Fuente},
  Nieto-Reyes, and Terán]{AliciaLuisTukey}
L.~{González-De La Fuente}, A.~Nieto-Reyes, and P.~Terán.
\newblock Statistical depth for fuzzy sets.
\newblock \emph{Fuzzy Sets and Systems}, 443:\penalty0 58--86, 2022.
\newblock ISSN 0165-0114.
\newblock \doi{10.1016/j.fss.2021.09.015}.
\newblock Fuzzy Intervals and Applications.

\bibitem[Horváth et~al.(2004)Horváth, Hušková, Kokoszka, and
  Steinebach]{Horvath2004}
L.~Horváth, M.~Hušková, P.~Kokoszka, and J.~Steinebach.
\newblock Monitoring changes in linear models.
\newblock \emph{Journal of Statistical Planning and Inference}, 126\penalty0
  (1):\penalty0 225--251, 2004.
\newblock ISSN 0378-3758.
\newblock \doi{10.1016/j.jspi.2003.07.014}.

\bibitem[Hubert et~al.(2015)Hubert, Rousseeuw, and Segaert]{Hubert2015}
M.~Hubert, P.~J. Rousseeuw, and P.~Segaert.
\newblock Multivariate functional outlier detection.
\newblock \emph{Statistical Methods {\&} Applications}, 24\penalty0
  (2):\penalty0 177--202, Jul 2015.
\newblock ISSN 1613-981X.
\newblock \doi{10.1007/s10260-015-0297-8}.

\bibitem[Jörnsten(2004)]{Joernsten2004}
R.~Jörnsten.
\newblock Clustering and classification based on the {$L_1$} data depth.
\newblock \emph{Journal of Multivariate Analysis}, 90\penalty0 (1):\penalty0
  67--89, 2004.
\newblock ISSN 0047-259X.
\newblock \doi{10.1016/j.jmva.2004.02.013}.
\newblock Special Issue on Multivariate Methods in Genomic Data Analysis.

\bibitem[Kaggle.com(2024)]{BananaData}
Kaggle.com.
\newblock Banana quality: Can you identify good bananas by their numerical
  characteristics?, 2024.
\newblock URL \url{https://www.kaggle.com/datasets/l3llff/banana}.
\newblock Uploaded by L3LLFF, retrieved April 5, 2024.

\bibitem[Lai(1976)]{Lai1976}
T.~L. Lai.
\newblock {On Confidence Sequences}.
\newblock \emph{The Annals of Statistics}, 4\penalty0 (2):\penalty0 265 -- 280,
  1976.
\newblock \doi{10.1214/aos/1176343406}.

\bibitem[Lan and DeMets(1983)]{LanDeMets1983}
K.~K.~G. Lan and D.~L. DeMets.
\newblock Discrete sequential boundaries for clinical trials.
\newblock \emph{Biometrika}, 70\penalty0 (3):\penalty0 659--663, 1983.
\newblock ISSN 00063444.
\newblock \doi{10.2307/2336502}.

\bibitem[Lange et~al.(2014)Lange, Mosler, and Mozharovskyi]{Lange2014}
T.~Lange, K.~Mosler, and P.~Mozharovskyi.
\newblock Fast nonparametric classification based on data depth.
\newblock \emph{Statistical Papers}, 55\penalty0 (1):\penalty0 49--69, Feb
  2014.
\newblock ISSN 1613-9798.
\newblock \doi{10.1007/s00362-012-0488-4}.

\bibitem[Li et~al.(2012)Li, Cuesta-Albertos, and Liu]{LiCuestaLiu2012}
J.~Li, J.~A. Cuesta-Albertos, and R.~Y. Liu.
\newblock {DD}-classifier: Nonparametric classification procedure based on
  {DD}-plot.
\newblock \emph{Journal of the American Statistical Association}, 107\penalty0
  (498):\penalty0 737--753, 2012.
\newblock \doi{10.1080/01621459.2012.688462}.

\bibitem[Liu et~al.(2008)Liu, Ting, and Zhou]{Liu2008}
F.~T. Liu, K.~M. Ting, and Z.-H. Zhou.
\newblock Isolation forest.
\newblock In \emph{2008 {E}ighth {IEEE} {I}nternational {C}onference on {D}ata
  {M}ining}, pages 413--422. IEEE, 2008.

\bibitem[Liu(1990)]{LiuSimplicial}
R.~Y. Liu.
\newblock {On a Notion of Data Depth Based on Random Simplices}.
\newblock \emph{The Annals of Statistics}, 18\penalty0 (1):\penalty0 405 --
  414, 1990.
\newblock \doi{10.1214/aos/1176347507}.

\bibitem[Liu and Singh(1993)]{LiuSingh}
R.~Y. Liu and K.~Singh.
\newblock A {Q}uality {I}ndex {B}ased on {D}ata {D}epth and {M}ultivariate
  {R}ank {T}ests.
\newblock \emph{Journal of the American Statistical Association}, 88\penalty0
  (421):\penalty0 252--260, 1993.
\newblock ISSN 01621459.
\newblock \doi{10.2307/2290720}.

\bibitem[Liu and Modarres(2011)]{LiuModarres2011}
Z.~Liu and R.~Modarres.
\newblock Lens data depth and median.
\newblock \emph{Journal of Nonparametric Statistics}, 23\penalty0 (4):\penalty0
  1063--1074, 2011.
\newblock \doi{10.1080/10485252.2011.584621}.

\bibitem[L{\'o}pez-Pintado and Romo(2009)]{LR2009}
S.~L{\'o}pez-Pintado and J.~Romo.
\newblock On the concept of depth for functional data.
\newblock \emph{Journal of the American Statistical Association}, 104\penalty0
  (486):\penalty0 718--734, 2009.
\newblock ISSN 01621459.
\newblock \doi{10.1198/jasa.2009.0108}.

\bibitem[Marsaglia(1972)]{Marsaglia1972}
G.~Marsaglia.
\newblock {Choosing a Point from the Surface of a Sphere}.
\newblock \emph{The Annals of Mathematical Statistics}, 43\penalty0
  (2):\penalty0 645 -- 646, 1972.
\newblock \doi{10.1214/aoms/1177692644}.

\bibitem[Mosler(2013)]{Mosler2013}
K.~Mosler.
\newblock Depth statistics.
\newblock In C.~Becker, R.~Fried, and S.~Kuhnt, editors, \emph{Robustness and
  Complex Data Structures: Festschrift in Honour of Ursula Gather}, pages
  17--34, Berlin, Heidelberg, 2013. Springer Berlin Heidelberg.
\newblock ISBN 978-3-642-35494-6.
\newblock \doi{10.1007/978-3-642-35494-6_2}.

\bibitem[Mozharovskyi and
  Valla(2025)]{mozharovskyi2024anomalydetectionusingdata}
P.~Mozharovskyi and R.~Valla.
\newblock Anomaly detection using data depth: multivariate case.
\newblock \emph{International Journal of Data Science and Analytics}, May 2025.
\newblock ISSN 2364-4168.
\newblock \doi{10.1007/s41060-025-00784-1}.

\bibitem[Nieto-Reyes and Battey(2016)]{Alicia}
A.~Nieto-Reyes and H.~Battey.
\newblock A {T}opologically {V}alid {D}efinition of {D}epth for {F}unctional
  {D}ata.
\newblock \emph{Statistical Science}, 31\penalty0 (1):\penalty0 61--79, 02
  2016.
\newblock \doi{10.1214/15-STS532}.

\bibitem[Nieto-Reyes and Battey(2021)]{AliciaHeatherConstruction}
A.~Nieto-Reyes and H.~Battey.
\newblock A topologically valid construction of depth for functional data.
\newblock \emph{Journal of Multivariate Analysis}, 184:\penalty0 104738, 2021.
\newblock ISSN 0047-259X.
\newblock \doi{10.1016/j.jmva.2021.104738}.

\bibitem[Nieto-Reyes and Cabrera(2022)]{Alicia2022}
A.~Nieto-Reyes and J.~Cabrera.
\newblock Statistical depth based normalization and outlier detection of gene
  expression data.
\newblock \emph{arXiv preprint arXiv:2206.13928}, 2022.
\newblock \doi{10.48550/arXiv.2206.13928}.

\bibitem[Oja(1983)]{Oja1983}
H.~Oja.
\newblock Descriptive statistics for multivariate distributions.
\newblock \emph{Statistics \& Probability Letters}, 1\penalty0 (6):\penalty0
  327--332, 1983.
\newblock ISSN 0167-7152.
\newblock \doi{/10.1016/0167-7152(83)90054-8}.

\bibitem[Ramsay et~al.(2019)Ramsay, Durocher, and Leblanc]{Ramsay2019}
K.~Ramsay, S.~Durocher, and A.~Leblanc.
\newblock Integrated rank-weighted depth.
\newblock \emph{Journal of Multivariate Analysis}, 173:\penalty0 51--69, 2019.
\newblock ISSN 0047-259X.
\newblock \doi{10.1016/j.jmva.2019.02.001}.

\bibitem[Robbins(1970)]{Robbins1970}
H.~Robbins.
\newblock {Statistical Methods Related to the Law of the Iterated Logarithm}.
\newblock \emph{The Annals of Mathematical Statistics}, 41\penalty0
  (5):\penalty0 1397 -- 1409, 1970.
\newblock \doi{10.1214/aoms/1177696786}.

\bibitem[Serfling(2002)]{SerflingSpatial2002}
R.~Serfling.
\newblock A depth function and a scale curve based on spatial quantiles.
\newblock In Y.~Dodge, editor, \emph{Statistical Data Analysis Based on the
  L1-Norm and Related Methods}, pages 25--38, Basel, 2002. Birkh{\"a}user
  Basel.
\newblock ISBN 978-3-0348-8201-9.

\bibitem[Serfling and Zuo(2000)]{ZuoSerfling}
R.~Serfling and Y.~Zuo.
\newblock {General notions of statistical depth function}.
\newblock \emph{The Annals of Statistics}, 28\penalty0 (2):\penalty0 461 --
  482, 2000.
\newblock \doi{10.1214/aos/1016218226}.

\bibitem[Sguera et~al.(2014)Sguera, Galeano, and Lillo]{Sguera2014}
C.~Sguera, P.~Galeano, and R.~Lillo.
\newblock Spatial depth-based classification for functional data.
\newblock \emph{TEST}, 23\penalty0 (4):\penalty0 725--750, Dec 2014.
\newblock ISSN 1863-8260.
\newblock \doi{10.1007/s11749-014-0379-1}.

\bibitem[Tan et~al.(2011)Tan, Ting, and Liu]{Tan2011}
S.~C. Tan, K.~M. Ting, and T.~F. Liu.
\newblock Fast anomaly detection for streaming data.
\newblock In \emph{{IJCAI} proceedings-international joint conference on
  artificial intelligence}, volume~22, page 1511. Citeseer, 2011.

\bibitem[Ting et~al.(2010)Ting, Zhou, Liu, and Tan]{Ting2010}
K.~M. Ting, G.-T. Zhou, F.~T. Liu, and J.~S.~C. Tan.
\newblock Mass estimation and its applications.
\newblock In \emph{Proceedings of the 16th ACM SIGKDD international conference
  on Knowledge discovery and data mining}, pages 989--998, 2010.

\bibitem[Tukey(1974)]{Tukey}
J.~W. Tukey.
\newblock Mathematics and the picturing of data.
\newblock In R.~D. James, editor, \emph{Proceedings of the International
  Congress of Mathematicians}, volume~2, pages 523--531, Vancouver, 1974.
  Canadian Mathematical Congress.

\bibitem[Williams et~al.(2002)Williams, Baxter, He, Hawkins, and
  Gu]{Williams2002}
G.~Williams, R.~Baxter, H.~He, S.~Hawkins, and L.~Gu.
\newblock A comparative study of rnn for outlier detection in data mining.
\newblock In \emph{2002 {IEEE} International Conference on Data Mining, 2002.
  Proceedings.}, pages 709--712, 2002.
\newblock \doi{10.1109/ICDM.2002.1184035}.

\bibitem[Yamanishi et~al.(2004)Yamanishi, Takeuchi, Williams, and
  Milne]{Yamanishi2004}
K.~Yamanishi, J.-i. Takeuchi, G.~Williams, and P.~Milne.
\newblock On-line unsupervised outlier detection using finite mixtures with
  discounting learning algorithms.
\newblock \emph{Data Mining and Knowledge Discovery}, 8\penalty0 (3):\penalty0
  275--300, May 2004.
\newblock ISSN 1573-756X.
\newblock \doi{10.1023/B:DAMI.0000023676.72185.7c}.

\bibitem[Yang and Modarres(2018)]{yang2018beta}
M.~Yang and R.~Modarres.
\newblock $\beta$-skeleton depth functions and medians.
\newblock \emph{Communications in Statistics-Theory and Methods}, 47\penalty0
  (20):\penalty0 5127--5143, 2018.
\newblock \doi{10.1080/03610926.2017.1386320}.

\bibitem[Zhao et~al.(2024)Zhao, Xu, Mu, Yang, and Wu]{Zhao2024}
W.~Zhao, Z.~Xu, Y.~Mu, Y.~Yang, and W.~Wu.
\newblock Model-based statistical depth with applications to functional data.
\newblock \emph{Journal of Nonparametric Statistics}, 36\penalty0 (2):\penalty0
  313--356, 2024.
\newblock \doi{10.1080/10485252.2023.2226262}.

\end{thebibliography}

\appendix
\section{Examples of E-depth functions}\label{appendix_depths}
\subsection{Multivariate Type A depths with indicator kernel}\label{subsec_multivariate_depths_A}
Let $\mathbb{S}:=\mathbb{R}^d$. For the simplicial depth \citep{LiuSimplicial}, in \eqref{DA} $r:=d+1$ and
\begin{align}\label{hs}
G^{S}(\bz,[\bm{\eta}_1,\ldots,\bm{\eta}_{d+1}])=\mathds 1_{\left\{\bz \in \mathrm{conv}\left\{\bm{\eta}_{1}, \ldots,\bm{\eta}_{d+1}\right\}\right\}},
\end{align}
where $\mathrm{conv}$ stands for the convex hull. The condition $\bz$ lies in $\mathrm{conv}\{\bm{\eta}_{1}, \ldots,\bm{\eta}_{{d+1}}\}$ is equivalent to checking whether a vector $\bm{\theta} \in \mathbb R^{d+1}$ with entries $\theta_j\geq 0$ for $j=1,\ldots, d+1$ and $ \sum_{j=1}^{d+1} \theta_j = 1$ exists such that $ \sum_{j=1}^{d+1} \theta_j \cdot \bm{\eta}_{j} = \bz$, which is implemented in standard software by checking if the following linear program has a feasible solution: 
\begin{align*}
	&\min_{\bm{\theta} \in \mathbb R^{d+1}} \quad  \langle \bm{1}, \bm \theta \rangle, \quad\text{with } \bm{1}=(1, \ldots,1)^T,\\
    &\text{s.t.} \quad  \sum_{j=1}^{d+1} \theta_j \cdot \bm{\eta}_{j} = \bz, \quad \sum_{j=1}^{d+1} \theta_j = 1,\quad \theta_j \geq 0 \text{ for each } j =1,\ldots,d+1.
\end{align*}
The  empirical simplicial depth given $\mathbb{R}^d$-valued random vectors $\bZ_1,\ldots,\bZ_n$, i.e.\ the standard estimator as in \eqref{typeAdepth}  of the theoretical simplicial depth, is given by
\begin{align}
    D_A(\bz;\bZ_1, \ldots,\bZ_n) = {n \choose d+1}^{-1} \sum_{1 \leq i_1<...<i_{d+1} \leq n} 
    \mathds 1_{\left\{\bz \in \mathrm{conv}\left\{\bZ_{i_1}, \ldots,\bZ_{i_{d+1}}\right\}\right\}}.
\end{align} 
    
Other well-known examples take in in \eqref{DA} $r:=2$, independently of $d$. They are: spherical depth \citep{elmore2006},  with 
\begin{align}\label{he}
    G^{E}(\bz;[\bm{\eta}_1,\bm{\eta}_2])= \mathds 1_{\left\{ (\bm{\eta}_1-\bz)^T \cdot (\bm{\eta}_2-\bz) \leq 0 \right\}},
\end{align}
lens depth \citep{LiuModarres2011}, with 
\begin{align}\label{hl}
    G^{L}(\bz;[\bm{\eta}_1,\bm{\eta}_2])= \mathds 1_{\left\{ \|\bm{\eta}_1 - \bm{\eta}_2\| \geq \max(\|\bz - \bm{\eta}_1\|,\|\bz - \bm{\eta}_2\|)\right\}}
\end{align}
and $\beta$-skeleton depth \citep{yang2018beta}, with 
\begin{align}\label{hb}
    G^{\beta}(\bz;[\bm{\eta}_1,\bm{\eta}_2])= \mathds 1_{\left\{ \|\bm{\eta}_1 - \bm{\eta}_2\| \geq \max\{\|\bm{\eta}_i+(\frac{2}{\beta}-1)\bm{\eta}_j-\frac{2}{\beta}\bz\|: \; i,j=1,2, i\neq j\}\right\}}.
\end{align}
Unlike  \eqref{hs}, conditions  \eqref{he}-\eqref{hb} are computationally easy to check.

\subsection{Functional Type A depths}\label{subsec_functional_depth_A}

Let $I \subseteq \mathbb R$ be a compact interval. Let $\mathbb{S}$ be an infinite dimensional Hilbert space on $\mathcal V$, with $\bz(t)\in \mathbb R$ for each $\bz\in\mathbb{S}$ and $t\in \mathcal V$. Let $\| \cdot\|_{L^2}$ be the $L^2$ norm and $r:=1$. For the h-depth \citep{Cuevas2007Projection}, we have
\begin{align}G^{\mbox{h}}(\bz;\bm{\eta}) =\frac{1}{\mbox{h}}\,K \left(\frac{\| \bz-\bm{\eta}\|_{L^2}}{\mbox{h}}\right),
\end{align}
with $K$ the Gaussian kernel. The bandwidth parameter $h$ can be a prefixed constant or data-driven  \citep{AliciaHeatherConstruction}.

Let $\mathbb{S}$ be the space of continuous functions on $\mathcal V$, with $\bz(t)\in \mathbb R$ for each $\bz\in\mathbb{S}$ and $t\in \mathcal V$. Let $r:=2$.
Making use of 2 curves-bands, for the band depth  \citep{LR2009}, we have
\begin{align}\label{band}
    G^{B}(\bz;[\bm{\eta}_1,\bm{\eta}_2]) =\mathds 1_{\left\{ \min( {\bm{\eta}}_1(t),{\bm{\eta}}_2(t)) \leq {\bz}(t)  \leq \max( {\bm{\eta}}_1(t),{\bm{\eta}}_2(t)) \mbox{ for all } t \in \mathcal V \right\}}
\end{align}
and for the modified band depth \citep{LR2009}
\begin{align}
    G^{M}(\bz;[\bm{\eta}_1,\bm{\eta}_2]) =\frac{ \lambda\left(  \left\{t \in \mathcal V: \min( {\bm{\eta}}_1(t),{\bm{\eta}}_2(t)) \leq {\bz}(t)  \leq \max( {\bm{\eta}}_1(t),{\bm{\eta}}_2(t)) \right\}\right)}{ \lambda(\mathcal V)},\label{eq_kern_mod_band}
\end{align}
where  $\lambda(\cdot)$ denotes the Lebesgue measure. The 2 curves-bands case is the most commonly used special case of a more general class based on $J$ curves-bands.  The seMCD-methodology discussed in this paper only  works for $J=2$. 

These depth functions can also be applied to multivariate data.

\subsection{Integrated depths}\label{subsec_IRW_depth}
Let $\mathbb{S}$ be a Banach space with a separable dual $\mathbb{S}^*$ and $\bm \eta$ a random element with probability measure on $\mathbb{S}^*$, for instance a non-degenerate Gaussian measure if $\mathbb{S}$ is infinite-dimensional or the uniform distribution on the unit sphere if it is finite-dimensional. Let $\bz\in \mathbb{S}$.  For the  integrated dual depth \citep{CuevasFraimanIntegrated2009} with respect to a distribution $P$, we have in \eqref{DI}
\begin{align}\label{dual}
    I^D(\bz; \bm{\eta}) = D_u( u(  \bm{\eta},\bz), u(\bm{\eta},\bZ)),
\end{align}
where $\bZ\sim P$, the quantity $D_u$ is a univariate depth function  and $u$ is a function related to it such that, conditioned on $\bm \eta$, it holds $u(  \bm{\eta},\bz)\in\mathbb{R}$ and $u(\bm{\eta},\bZ)$ is a random variable. 
 
In the particular case that $\mathbb{S}=R^d$, $\bm \eta$ is a random element with the uniform distribution on the unit sphere in $\mathbb R^d$  and $D_u$ is the univariate Tukey depth \citep{Tukey}, we have the integrated rank-weighted (IRW) depth \citep{Ramsay2019}. Thus, \eqref{dual} results in
\begin{align}\label{T}
        I^R(\bz; \bm{\eta}) = \min\left\{
        \P(  \bm{\eta}^T \cdot \bZ \leq  \bm{\eta}^T \cdot \bz | \bm{\eta}), \P(  \bm{\eta}^T \cdot \bZ \geq  \bm{\eta}^T \cdot \bz | \bm{\eta})  \right\}.
\end{align} 
Using $\bZ_1,\ldots,\bZ_n\overset{\textrm{iid}}{\sim} P$, we can write an empirical version of $\P( \bm{\eta}^T \cdot \bZ \leq  \bm{\eta}^T \cdot \bz|\bm{\eta})$ as
\begin{align*} 
   \frac{1}{n}\sum_{j=1}^n\mathds 1_{\{
   \bm{\eta}^T \cdot \bZ_j \leq  \bm{\eta}^T \cdot \bz \}}
\end{align*} 
and analogously for the other probability in \eqref{T}. Thus, approximating an empirical version of the IRW depth simply requires to also do a numerical approximation of the expectation.  Additionally, $\bm{\eta} \deq \bm W/\|\bm W\|$ holds with $\bm W$ being a standard normal distribution on $\mathbb R^d$ \citep{Marsaglia1972}. Then, because of the inequalities in \eqref{T}, it is equivalent to take $\bm{\eta}$ distributed as a  standard normal on $\mathbb R^d$. Thus, the empirical IRW depth can easily be approximated by Monte Carlo methods, fitting the seMCD-framework.
\end{document}